\documentclass[letterpaper,10pt,conference]{IEEEtran}

\renewcommand{\baselinestretch}{0.993}

\usepackage{url}
\usepackage{cite}
\usepackage{tikz}
\usepackage{stmaryrd}
\usepackage{amssymb}
\usepackage{amsthm}
\usepackage{framed}
\usepackage{listings}
\usepackage{diagbox}
\usepackage{bm}
\usepackage{makecell}
\usepackage{multirow}
\usepackage{threeparttable}
\usepackage{todonotes}
\usepackage{mdwlist}
\usepackage{float}
\usepackage{microtype}
\usepackage{mathpartir}
\usepackage{mathrsfs} % police mathscr
\usepackage{xcolor}
\usepackage[colorlinks=true, citecolor=blue, pdfborder={0 0 0}]{hyperref}
\usetikzlibrary{decorations.pathreplacing, decorations.pathmorphing, shapes.misc} %brace
\usetikzlibrary{arrows}
\usetikzlibrary{graphs}
\usetikzlibrary{positioning,calc,shapes} 
\DeclareRobustCommand{\thinskip}{\hskip 0.16667em\relax}
\def\emdash{---}
\def\d@sh#1#2{\unskip#1\thinskip#2\thinskip\ignorespaces}
\def\Dash{\d@sh\nobreak\emdash}
\def\Ldash{\d@sh\empty{\hbox{\emdash}\nobreak}}
\def\Rdash{\d@sh\nobreak\emdash}

\usepackage{enumitem}

\usepackage{tikz}

\definecolor{vertf}{RGB}{0,153,51}
\definecolor{rougef}{RGB}{200,0,0}
\RequirePackage{pifont}
\newcommand{\tick}{\textcolor{vertf}{\ding{51}}}
\newcommand{\wrongtick}{\textcolor{rougef}{\ding{55}}}

\newcommand{\asterios}{AnonymOS}

\newcommand{\multipleinstruction}{\leadsto}
\newcommand{\multipleinstructionattacker}{\stackrel{A}{\leadsto}}

\newcommand{\longversion}[1]{#1}
\renewcommand{\longversion}[1]{}

\newtheorem{definition}{Definition}
\newtheorem{theorem}{Theorem}
\newtheorem{corollary}{Corollary}
\newtheorem{assumption}{Assumption}

\definecolor{bluekeywords}{rgb}{0.13, 0.13, 1}
\definecolor{greencomments}{rgb}{0, 0.5, 0}
\definecolor{redstrings}{rgb}{0.9, 0, 0}
\definecolor{graynumbers}{rgb}{0.5, 0.5, 0.5}

\usepackage{amsmath}
\author{\IEEEauthorblockN{Olivier Nicole\IEEEauthorrefmark{1}\IEEEauthorrefmark{2},
Matthieu Lemerre\IEEEauthorrefmark{1},
S\'ebastien Bardin\IEEEauthorrefmark{1} and
Xavier Rival\IEEEauthorrefmark{2}\IEEEauthorrefmark{3}}\\
\IEEEauthorblockA{\IEEEauthorrefmark{1}CEA List, Software Reliability and Security Laboratory, Paris-Saclay, France} %Software Reliability and Security Laboratory,P.C. 174, Gif-sur-Yvette, 91191, France2
\IEEEauthorblockA{\IEEEauthorrefmark{2}D\'epartement d'informatique de l'ENS, ENS, CNRS, PSL University, Paris, France}
\IEEEauthorblockA{\IEEEauthorrefmark{3}INRIA, Paris, France}\\
olivier.nicole@cea.fr, matthieu.lemerre@cea.fr, sebastien.bardin@cea.fr, xavier.rival@ens.fr}

\date{\today}
\title{Automatically Proving Microkernels Free from Privilege Escalation from their Executable}

\newcommand{\myparagraph}[1]{\smallskip\subsubsection{#1}}
\newcommand{\mysubparagraph}[1]{\smallskip\noindent\textbf{#1}.}
\newcommand{\mysubparagraphbis}[1]{\smallskip\noindent\textbf{#1}}
\newcommand{\binseccodex}{\mbox{\textsc{BINSEC/Codex}}}
\newcommand{\binsec}{\textsc{BINSEC}}
\begin{document}

\maketitle

\newcommand{\xxx}[1]{\textcolor{red}{XXX: #1}}
\newcommand{\maybe}[1]{\textcolor{green}{MAYBE: #1}}
\newcommand{\maybenot}[1]{}

\begin{abstract}
  Operating system kernels are the security keystone of most computer
  systems,  as they provide  the core protection
  mechanisms. Kernels are in particular responsible for their own security, 
  i.e.~they must prevent untrusted user tasks from reaching their
  level of privilege. 
  We demonstrate that proving such \emph{absence
    of privilege escalation} is a pre-requisite for any definitive
  security proof of the kernel. 
% : every system invariant must imply this
%  property, and is thus necessarily large and complex. 
% 
 While prior OS kernel formal verifications were performed either on source code 
or crafted kernels, with manual or semi-automated methods requiring  significant human efforts in annotations 
or proofs,    
we show that it is possible to compute such kernel security proofs using  % an invariant using
  \emph{fully-automated} methods and starting from the \emph{executable code} of an
    \emph{existing microkernel with no modification}, thus formally
  verifying absence of privilege escalation with high confidence for
  a low cost. We applied our method on two embedded microkernels, including the industrial kernel AnonymOS\footnote{\label{anonimization}{The name was anonymized for the sake of double-blind review.}}: with only 58 lines of annotation and less than 10 minutes of computation, our
  method finds a vulnerability in a first (buggy) version of AnonymOS and verifies absence of
  privilege escalation in a second (secure)~version. 

% This should enable reaching a \emph{high level of assurance with a lower cost} than previous methods.

%   \xxx{Et insister sur le caractère définitif d'une preuve. Si je
%     vérifie une propriete sur un bout de code C, cette preuve n'est
%     definitive que si les hypotheses que je fait (e.g. sur l'état
%     mémoire initial, sur le fait que le code est bien là...) sont
%     vérifiées.  On fait une preuve definitive, donc on n'a plus aucune
%     hypothèse sur le code à la fin}
  
\end{abstract}

%\myparagraph{}

\section{Introduction}

\mysubparagraph{Context} The security of  many computer systems 
builds upon that of its operating system
kernel.  %
We define a kernel as \emph{a computer program that prevents
untrusted code from performing arbitrary actions}, in particular
performing arbitrary hardware and memory accesses\footnote{This definition includes
  kernels and hypervisors with software-based \cite{bershad1995spin,yang2010safe,hunt2007singularity} or hardware-based isolation, but excludes the hardware abstraction layers of operating systems without memory protection~\cite{dunkels2004contiki,levis2005tinyos}.%, FreeRTOS before v6).
}. Kernels that fail to prevent this are said
to be vulnerable to \emph{privilege escalation} attacks\footnote{Many systems are still vulnerable to such attacks: more than 240 related CVEs have been issued for the Linux kernel,  
%for the ``execute code'' vulnerability alone. 
 and embedded systems are also impacted\Dash{}see VxWorks (CVE-2019-9865)  and POK OS~\cite{duverger2019gustave}.}.  
As
this vulnerability is of the highest severity and can affect any
kernel, formally verifying that kernels cannot be exploited by
untrusted code to gain access to hardware privilege is of the uttermost importance. 
 Actually, \emph{proving absence of privilege escalation (APE) is a
  mandatory  step when attempting to formally verify a kernel:
  nothing can be proven \emph{unconditionally} about a kernel unless
  this property holds}. 
% we demonstrate (Theorem~\ref{th:priv-escal-state-invariant}) that

\mysubparagraph{Scope} Besides large well-known monolithic kernels
(e.g., Linux, Windows, *BSD) whose size and complexity are currently out
of reach of formal verification, there is a rich ecosystem of
small-size kernels found in many industrial applications, such as
security or safety-critical applications, embedded or IoT
systems. This includes security-oriented kernels like separation
kernels \cite{rushby1981design}, microkernels \cite{klein2009sel4},
exokernels \cite{engler95exokernel} and security-oriented hypervisors
\cite{vasudevan2013design,vasudevan2016uberspark} or enclave software
\cite{ferraiuolo2017komodo}; but also kernels used in embedded
systems, for example in microcontrollers \cite{levy17tock}, real-time
\cite{ramamritham1994scheduling,muehlberg2011verifying} or
safety-critical \cite{ARINC653,richards2010modeling} operating
systems.

% We focus on secure kernels, that should comply with security
% principles (least common mechanism, economy of mechanism) and should
% thus be small in size and complexity~\cite{saltzer73protection}.

% Formal verification requires too much effort for large monolithic
% kernels, but is useful and feasible in kernels of smaller sizes; this
% includes security-oriented kernels like separation kernels
% \cite{rushby1981design}, microkernels \cite{klein2009sel4}, exokernels
% \cite{engler95exokernel} and security-oriented hypervisors
% \cite{xmhf,hyper-v?}; but also operating systems for embedded systems,
% like real-time operating systems or operating systems for
% safety-critical systems.

\emph{We focus on such small-size kernels}. To have a practical impact on
these systems, a formal verification must be:

\begin{itemize}
\item {\bf Non-invasive}: the verification should be applicable to the
  kernel as it is. Formal verification methods that require heavy
  annotations or rewrite in a new language are highly expensive  and
  require developers with a rare combination of expertises (OS 
  %implementation 
   and formal methods);
  %%%%% Moreover, the need to rewrite their kernel from scratch
  %%%%% and to use other tools than their standard toolchains deter OS
  %%%%% developpers from using formal verification.
  % The formal verification effort should be kept separated from
  % the OS development effort.

\item {\bf Automated}: the cost and effort necessary to perform a
  formal verification should be minimized. Formal verification
  techniques that are {\it manual} (e.g., proof assistants) or {\it semi-automated} 
  (e.g., deductive verification) require a large proof or annotation
  effort;
  % . Moreover, a part of this work has to
  % be redone when the kernel evolves. % Therefore the amount of
  % % automation should be maximized; 
  %
  %The only required effort should be to specify the properties to be
  %verified and to setup a small amount of analysis parameters;

\item {\bf Close to the running system}: Verification should be performed on the machine
  code~\cite{bevier1989kit,klein2014comprehensive} in order to remove the whole build chain (compiler,
  linker, compilation options, etc.) from the trust base.
  % . Indeed, verifying only
  % the source code requires to trust the whole build chain (compiler,
  % linker, makefiles etc.).
  % and to rely on hypotheses (e.g. absence of code
  % automodification, well-formedness of the stack) that cannot be
  % verified using the source code alone.
  Machine code verification is all the more important
  on kernels, as they contain many error-prone low-level
  interactions with the hardware, not described by the source-level
  semantics.%  (e.g., use of special registers, interaction with assembly
  % code, explicit stack manipulation, memory translation, location of
  % the code in memory, etc.).

% \item {\bf Close to the running system}: to provide the highest level
%   of assurance, the verification should be performed on the software
%   that will be run, like a nonvolatile memory dump or the executable
%   file to be loaded by a trusted bootloader. Verification on the
%   machine code is especially important in the case of OS kernels: they
%   contain many error-prone, low-level interactions with the hardware,
%   that are not described by the high-level semantics of source code
%   (e.g., use of special registers, interaction with assembly code like
%   calling conventions, explicit manipulation of the stack, memory
%   translation, location of the code in memory). Verifying a property
%   on the source code only requires 1. to trust a large amount of code
%   transformation software (like compilers, linker scripts,
%   makefiles...) and 2. to rely on hypotheses (e.g. absence of code
%   automodification, well-formedness of the stack) that cannot be
%   verified using the source code alone.
\end{itemize}

Despite significant advances in the last decades
\cite{walker1980specification,bevier1989kit,klein2009sel4,yang2010safe,gu2015deep,xu2016practical,alkassar2010automated,dam2013machine,vasudevan2016uberspark,ferraiuolo2017komodo,paul2012completing,sewell2013translation}, 
existing kernel verification methods do not address these issues. In
most cases, verification is applied to microkernels developed or
rewritten for the purpose of formal verification (except \cite{xu2016practical}), and is performed only on 
source~\cite{gu2015deep,xu2016practical} or assembly~\cite{yang2010safe,paul2012completing,sewell2013translation,vasudevan2016uberspark}, 
 using highly expensive 
manual~\cite{bevier1989kit,klein2009sel4,gu2015deep,xu2016practical}
or
semi-automated~\cite{alkassar2010automated,yang2010safe,dam2013machine,vasudevan2016uberspark,ferraiuolo2017komodo}
methods. For example, the functional verification of the SeL4 microkernel~\cite{klein2009sel4} required 200,000 lines of annotations and still left 
parts of the code unchecked (boot, assembly).
% \todo[inline]{Because of the high cost and level of expertise,
% formal verification is currently out of the reach of mid-size
% businesses like AnonymFirm, the company behing AnonymOS.}

%  In fact, it is generally believed that verifying high-level
% properties such as absence of privilege escalation cannot be done
% using only automatic methods \cite{klein2009sel4}

\mysubparagraph{Goal and challenges} We focus on the key property of privilege escalation, and seek to \emph{design a
  fully-automated program analysis able to prove the absence of
  privilege escalation in microkernels for embedded systems, directly
  from their executable}.  Besides the well-documented difficulty of
static analysis of machine code~\cite{reps2010there}, solving this goal poses two main
technical challenges:

\begin{itemize}

\item Automatically proving absence of privilege escalation requires a formal 
  definition that is generic (i.e. independent of the kernel)  and suitable 
  to the  machine-code level. %
  %%%%% . To comply with our goal, this definition should not
  %%%%% require writing complex annotations requiring in-depth knowledge of
  %%%%% the source code of the kernel.
  Indeed, absence of privilege escalation is usually established through higher-level properties such as 
   safety of  control-flow and  memory together with preservation of invariants on protection mechanisms~\cite{vasudevan2016uberspark}, %\todo{Find citation},  
 % memory table invariants  
but  each of these properties requires an
  in-depth specification of the kernel behavior and knowledge of its
  source code\Dash{}preventing automated machine-level verification. In addition this definition should be amenable to static analysis, preferably using standard techniques; 
  
\item Most  kernels are \emph{parameterized} systems 
  designed to run an \emph{arbitrary} number of tasks. This is also
  true for microkernels  in embedded systems: even if the number
  of tasks, size of scheduling tables and communication buffers often
  do not vary during  execution, they depend on the application using the kernel. 
%  that makes use of the microkernel. 
  A flat representation of memory (enumerating all memory cells) \cite{dam2013machine} is no longer sufficient 
  in such a setting, and  
%  enumerate the properties of each memory cell 
   we need more
complex representations able to precisely \emph{summarize} memory, like
%  requires using precise abstractions that summarizes the memory, i.e.
  {\it shape abstract domains}
  \cite{sagiv1999parametric,chang2008relational}. Unfortunately, they usually    %usually
  require a large amount of manual annotations,
  which defeats our goal  of automation.
\end{itemize}

% \noindent No existing work targets verification of absence of
% privilege escalation. The idea of formally verifying a kernel directly
% from machine code was already explored, \cite{bevier1989kit,dam2013formal},
% but was applied to specially-designed kernels manually
% annotated with kernel and loop invariants, and did not address the
% challenge of parametrization (handling only a fixed number of partitions).

% \textbf{Old goal}

% Our goal is to show that fully-automated methods are mature enough and can
% be used to verify complex properties such as absence of privilege
% escalation in an OS microkernel. For this we used a maximalist approach
% (Section~\ref{sec:case-study}): we applied \emph{only} an automated
% analysis (Section~\ref{sec:sound-static-analysis}), on an
% \emph{existing}, \emph{unmodified} \emph{executable file} of a real
% kernel used in the industry. Actually, we never saw the source of this
% kernel.

% In general, OS kernels are developed so that they can be used
% independently of the application that uses them; thus they are
% \emph{parameterized} systems. This required us to develop a
% \emph{formal model} of the system
% (Section~\ref{sec:formalization-formal-description}) suitable for
% automated verification of parameterized system software, that we also
% used to prove Theorem~\ref{th:priv-escal-state-invariant}.

%\todo{ON: BINSEC/Codex c'est le nom de l'outil, pas de l'analyse}
\mysubparagraph{Contributions} We propose \binseccodex{}, a
novel static analysis for proving absence of privilege
escalation in microkernels from their executable. Our contributions include: 

\begin{itemize}

\item An original \emph{formal model}
  (Section~\ref{sec:formalization-formal-description})  suitable for
  defining privilege escalation attacks on parameterized kernel code  and
  allowing to reduce the proof of absence of privilege escalation to a
  standard program analysis problem (finding {\it non-trivial state
  invariants}, {\bf Theorem}~\ref{th:invariant-implies-noescalation}),
  hence reusing the standard program analysis machinery. % for program analysis.
  We also  prove that absence of privilege escalation is the most
  fundamental kernel property, without which
  nothing can be proved~({\bf Theorem}~\ref{th:priv-escal-state-invariant});

\item A new \emph{3-step methodology} (Section~\ref{sec:general-method})
  for %inferring kernel invariants and 
  proving absence of privilege
  escalation of parameterized kernels from their executable,
  featuring 1. automated extraction of most of the analysis
  (shape) annotations from kernel types; 2.  
  parameterized fully-automated binary-level static analysis inferring an {\it invariant of 
  the kernel under  
   a precondition on the user tasks}, and 3. 
  fully-automated method to check that the user tasks satisfy the
  inferred~precondition;

\item A \emph{novel weak shape abstract domain}
  (Section~\ref{sec:shape-abstract-domain}) able to verify the
  preservation of memory properties in parameterized  kernels. This
  domain is \emph{efficient} % in terms of performance} 
  (thanks to a dual
  flow sensitive/flow insensitive representation),  \emph{easily 
    configurable} (based on the memory layout of C types, most of the annotations are
  extracted automatically), and \emph{suitable to machine code
    verification} (e.g., addressing indexing of data structures using
  numerical offsets);

%   A new static analysis for automated verification of parameterized
%   microkernels from their executable, featuring 1. a state-of-the-art
%   binary-level static analyzer, 2. an automatic extraction of shape
%   invariants and parameters checker, 3. a novel shape abstract
%   domain automatically parameterized by these invariants. Altogether,
%   this analysis takes as input the executable code and data, and the
%   types of the parameters; and outputs a safe invariant of the
%   reachable states of the system.
  
% \item A novel {\bf shape abstract domain} allowing precise automated
%   verification of parameterized systems, and whose
%   configuration can be almost entirely extracted from standard
%   debugging information.

\item A \emph{thorough evaluation} of our method on \emph{two
    different microkernels} (Section~\ref{sec:case-study}), using two
  different instruction sets (\emph{x86} and \emph{ARMv7}) and memory protection
  mechanisms (\emph{segmentation} and \emph{pagination}). This includes the study of
  an ARM Cortex-A9 port of \emph{AnonymOS\footnote{AnonymOS is a concurrent industrial
    microkernel developed by AnonymFirm\textsuperscript{\ref{anonimization}}, a leading tool provider for
  safety-critical real-time systems, with presence in the
  aerospace, automative, and industrial automation markets.}}. The method is able to {\it find a vulnerability} in
  a beta version of this kernel, and to {\it verify the absence of
    privilege escalation} in a later version, in less than 450 seconds
  and with only 58 lines of manual annotations\Dash{}several order of
  magnitudes less than prior verification efforts
  (Table~\ref{tab:comparison-os-verification},
  p.~\pageref{tab:comparison-os-verification}).
%  and with a minimal trusted base.
\end{itemize}

% \todo[inline]{While EducRTOS is a small academic microkernel developed for teaching purpose, AnonymOS is an industrial microkernel
%   developed by AnonymFirm, a leading tool provider for
%   safety-critical real-time systems, with market presence in the
%   aerospace, automative, and industrial automation industries.}

\noindent This work is the \textbf{\textit{first}} OS verification effort  to specifically address 
\emph{absence of privilege escalation}.  It is also the
\textbf{\textit{first}} to perform formal verification on an {\it existing} operating
  system kernel \emph{without any modification}, on {\it machine code}, and the \textbf{\textit{first}} to do so using  a
\emph{fully-automated technique} able to handle \emph{parametrization}. Finally, it is the \textbf{\textit{first}} shape analysis performed on machine code.

% to verify properties of an
% operating system kernel \emph{using only a fully-automated verification
% technique}.  It is also the first to perform \emph{formal verification of an
% existing operating system kernel, in its entirety} (including all
% assembly code) and \emph{without any modification}.

We thus show that, contrary to a widespread
belief~\cite{klein2009sel4}, fully-automated methods like static analysis can be used to verify complex properties such as absence of
privilege escalation in embedded microkernels, directly from their executable.

\mysubparagraph{Limitations} %
% Our method is independent from the hardware architecture. Especially,
% our formalization is generic enough to encompass both MMU and MPU
% protection mechanisms, and the static analysis is independent from the
% instruction set  as it relies on an intermediate
% representation~\cite{bardin2011bincoa}.  Hence, while our case study
% is on ARMv7 architecture and MMU, the analysis should apply 
% similarly  on other hardware configurations.
Like any sound static analyzer, \binseccodex{} may be too imprecise on
some code patterns, emitting  {\it false alarms}. 
%and may not succeed in computing an invariant without emitting false alarms.
Currently, our analysis cannot handle   dynamic task spawning nor dynamic modification of memory
repartition, as well as self-modification or  code generation in the kernel.     
% deal with dynamic task
% spawning in full generality; yet it could handle dynamic spawning in
% the case where task structures are pre-allocated.
%Also, the analysis cannot cope with self-modification or code
%generation in the kernel. 
Still, many microkernels and
hypervisors fall in our scope~\cite{rushby1981design,dam2013formal,vasudevan2013design,vasudevan2016uberspark,xu2016practical,duverger2019gustave,ramamritham1994scheduling,muehlberg2011verifying,ARINC653,richards2010modeling}.

Finally, while absence of privilege escalation is arguably the most
important property of a kernel, verifying task separation is also of great
importance. This is left as future work.

% The effort necessary to perform a formal proof of a system, depends on
% its size and complexity. Our focus is on secure kernels, that should
% comply with security principles (least common mechanism, economy of
% mechanism \cite{saltzer73protection}) and should thus be small in size
% and complexity. Such secure kernels include for instance separation
% kernels, microkernels, exokernels and security-oriented
% hypervisors. However, other kernels of small size and complexity, as
% commonly found in embedded systems (e.g., operating systems for
% real-time or safety-critical operations) can also be targeted by our
% method. \maybe{Such a restriction allows for instance to run a
%   fully-context sensitive analysis, which would be too costly on a
%   large system.}

\section{Overview}
\label{sec:overview-motivating-example}

%We explain our method on  simplified hardware and kernel.

\subsection{System description, attacker model and trust base}
%\subsection{Setting}

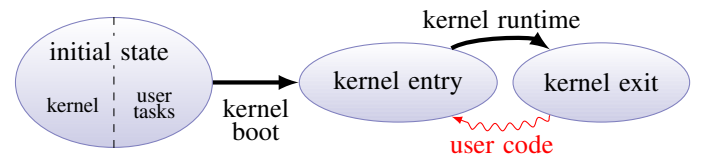
\begin{figure}[b]
  \centering\begin{tikzpicture}[xscale=2.7,
    main node/.style={draw=blue!50!black!50,top color=white,bottom color=blue!50!black!20}
    ]
%    \node[align=center,draw,ellipse] (init) at (0,0) {initial\\state};

      % \begin{scope}%[xshift=3mm]
      %   \draw[] (0,0) arc(0:180:4mm and 4mm) -- cycle;
      %   \draw[] (0,0) arc(0:-180:4mm and 4mm) -- cycle;
      %   \begin{scope}[xshift=-4mm,yshift=-1.7mm] \node at (0,0) {user tasks}; \end{scope}
      %   \begin{scope}[xshift=-4mm,yshift=2mm] \node at (0,0) {kernel}; \end{scope}        
      % \end{scope}

% %initial state
      
% %      \node[align=center,draw,ellipse split] (init) at (0,0) {kernel\nodepart{lower}{user tasks}};
%       \node[align=center,draw,ellipse] (kernel exit) at (1,0) {kernel exit};
%       \node[align=center,draw,ellipse] (kernel entry) at (2,0) {kernel entry};

%       \draw (0,0) edge[->, >=latex, ultra thick] node[below,align=center] {kernel\\boot} (kernel exit);
%       \draw (kernel exit) edge[red,decorate,decoration={snake,amplitude=.5mm,segment length=2mm,post length=2mm},bend left=60,->, >=latex] node[above] {user code} (kernel entry);
%       \draw (kernel entry) edge[bend left=60,->, >=latex, ultra thick] node[below,align=center] {kernel runtime} (kernel exit);

%    \draw (0,0) edge[->, >=latex, ultra thick] node[below,align=center] {\ \\kernel\\boot} (kernel entry);
    
    \node[main node,align=center,draw,ellipse,inner sep=2pt,minimum height=3em] (kernel exit) at (2,0) {kernel exit};
    \node[main node,align=center,draw,ellipse,inner sep=2pt,minimum height=3em] (kernel entry) at (1,0) {kernel entry};
    \node[main node,align=center,draw,ellipse,inner sep=2pt] (initial state) at (-0.4,0) {initial state\ \ \\ \\\phantom{user tasks}};

    \node[font=\footnotesize] at (-0.6,-0.3) {kernel};
    \node[align=center,font=\footnotesize] at (-0.2,-0.3) {user\\[-3pt]tasks};    

    % \draw (initial state.west) -- (initial state.east);
    % \draw (initial state.center) -- (initial state.south);

    % \draw (initial state.west) rectangle (initial state.south);
    % \draw (initial state.east) rectangle (initial state.south);

    % \draw[dashed] (initial state.north) -- (initial state.south);
    \draw[dashed] (initial state.north) -- +(0,-0.25);
    \draw[dashed] (initial state.south) -- +(0,1.1);    
    
%    \node[rectangle,draw] at (-0.4,0.4) {kernel};
%    \node[rectangle,draw] at (-0.4,-0.4) {user tasks};    

    \draw (initial state) edge[->, >=latex, ultra thick] node[below,align=center] {\\[-3pt]kernel\\[-3pt]boot} (kernel entry);
    \draw (kernel exit) edge[red,decorate,decoration={snake,amplitude=.5mm,segment length=2mm,post length=2mm},bend left=60,->, >=latex] node[below] {user code} (kernel entry);
    \draw (kernel entry) edge[bend left=60,->, >=latex, ultra thick] node[above,align=center] {kernel runtime} (kernel exit);        

    % \begin{scope}[shift={(0,-1.3)}]
    %   \draw[thick,green!70!black,decorate,decoration={brace,amplitude=3pt}] (2,0) -- (1,0);
    %   \draw[thick,red,decorate,decoration={brace,amplitude=3pt}] (1,0) -- (0,0);      
    %   \node[green!70!black] at (1,-0.3) {$\times$};
    %   \draw[thick,green!70!black] (0,-0.3) -- (1,-0.3);
    % \end{scope}

    % \begin{scope}[shift={(0,1.3)}]
    %   \node[green] at (0,0.3) {$\times$};
    %   \draw[thick,green!70!black,decorate,decoration={brace,amplitude=3pt}] (0,0) -- (2,0);
    % \end{scope}
    
  \end{tikzpicture}

 \caption{The typical system loop between kernel and user tasks.}
 \label{fig:system-control-flow}
 
\end{figure}

\mysubparagraph{System description} %
We consider a computer system consisting in \emph{hardware} running
both \emph{untrusted} software (\emph{user tasks})  and \emph{security-critical} software (including the \emph{kernel}). 
Our goal is to ensure that \emph{the
  only code running as privileged is the uncompromised,
  security-critical code}, where being privileged typically corresponds 
 to having a hardware flag set ({\it supervisor
mode}). Section~\ref{sec:formalization-formal-description} formalizes
these notions. 

\mysubparagraph{Attacker model} The attack goal is to escalate privilege, either by running untrusted software
with privilege or by injecting code into the security-critical software.

The attacker controls the user image\Ldash{}containing the user tasks code and data\Rdash{} loaded with the kernel before boot; can perform any software-based attack, such as modifying user task code and memory at runtime; but cannot make the hardware
deviate from its specification, and thus cannot perform physical
attacks nor exploit hardware backdoors or 
glitches~\cite{seaborn2015exploiting}.
% Protection against hardware
% attacks is orthogonal to our work.

\mysubparagraph{Trust base} We want to trust a minimal number of components:
\begin{itemize}
% \item We trust that our mathematical model faithfully represents the
%   hardware: e.g. if the hardware allows to write to kernel memory
%   using a DMA engine, the model must take this into account.
\item the software used to 
  load the kernel and user tasks in memory  (bootloader, EEPROM flasher, etc.), 
 
\item the tools used to perform the formal
  verification.  %  (i.e., the static analyzer and user tasks checker of
  % Section~\ref{sec:general-method})
  %including their model of the hardware.
\end{itemize}

\noindent Note that {\it we do not need to trust}  
%\begin{itemize}
%\item 
the kernel source code 
%\item 
nor any software  producing  the executable (e.g., compiler, assembler,
  linker, or build scripts).  
%\end{itemize}
%
%%%Instead, we perform an automated verification to ensure that the executable files built from these
%%%latter components  are not vulnerable to privilege escalation. 

%\longversion{This approach is independent from how the
%  kernel is developed: in our case study, we have verified the kernel
%  executable without even seeing its source code.}

% Instead, the executable files built from these latter components are
% verified by our analyzers to guarantee that privilege escalation is
% impossible. 

\subsection{Illustrative example of an OS kernel} \label{sec:toy}

%\subsubsection{Description of the system}

To better understand privilege escalation attacks, how a kernel is
structured to prevent these attacks, and how we verify  
that these attacks are impossible, we use, as an example, the barebone OS kernel of 
Figure~\ref{fig:a-tiny-kernel}. Although minimal, it   
handles  an arbitrary number of tasks and features  scheduling, memory protection  and
multiple levels of privilege. In particular, 
%if correctly implemented, 
it  should protect itself from the user tasks and be immune 
to privilege escalation attacks. The exemple is written in pseudo-C,
but remember that our analysis is performed on machine~code.

\mysubparagraph{Execution context of a task} This  kernel 
performs context switching between pre-defined tasks. Assuming a stack machine programming model, an
execution context is  defined here  by a read-only executable code
segment, a writable data segment holding the stack, a status register
\texttt{flags}, and two registers \texttt{pc} (program counter)
and \texttt{sp} (stack pointer) respectively pointing inside the
code and data segments. 

\mysubparagraph{Privilege level} %
A bit inside the  \texttt{flags} register %is a bitfield storing results of comparisons 
%(sign, equality, overflow flags) 
%and, more importantly,  
indicates  whether execution is \emph{privileged}.  Being
privileged allows executing instructions that would change the
privilege in the \texttt{flags} registers, or would change the values
of the \emph{system registers} (\texttt{mpu}$_1$, \texttt{mpu}$_2$,
\texttt{pc}$'$, \texttt{sp}$'$ and \texttt{flags}$'$), using $'$ to 
denote what is sometimes called  {\it banked} or {\it shadow} registers. %or \emph{shadow} registers.

\mysubparagraph{Memory protection} A running task can access only the memory inside its code and data segment. This
is enforced by the hardware through a Memory Protection Unit (MPU) 
controlled by two system registers \texttt{mpu}$_1$ and
\texttt{mpu}$_2$. Each \texttt{mpu} register enables access to the memory addresses
in a certain range  for reading, writing  or execution, depending
on the contents of the register. When the processor is in unprivileged
mode, all  other memory accesses are forbidden,  while 
in privileged mode  the MPU is bypassed and  the whole  memory can be accessed.  

\mysubparagraph{Kernel} %Most OSes rely on the use of the privileged
%execution levels provided by the processor. 
%
The \emph{kernel} is
the only program which is supposed to run privileged.  Kernel execution proceeds as follows:   
\begin{enumerate}

\item First, an \emph{interrupt} occurs, which is the only way to switch from unprivileged to privileged execution. The hardware
saves the context of the executing task and begins executing the kernel. In our example, the hardware: (a) saves the value of
 \texttt{pc}, \texttt{sp}, and \texttt{flags};  (b) restores the kernel stack pointer \texttt{sp} to the kernel stack and restores \texttt{flags} to allow privileged execution; 
  and (c) sets  \texttt{pc} to begin executing the
\texttt{kernel} function;

\item Then, the kernel dispatches the execution according to the
  interrupt received.  The only special case is the
  \texttt{RESET} interrupt,  % --  called when the system starts or is reset, 
   which  boots (i.e. initializes) the system.
  Here, it  consists in setting
  \texttt{cur}, a variable always pointing to the task that has been
  or will be executed. For all other interrupts the system is already
  executing a task, and its execution context must be saved in memory
  before we can switch to another task; 

\item After that, the kernel chooses another task to be scheduled\Dash{}here with a  simple round-robin method, 
so it is sufficient  to follow the 
\texttt{next} field in a circular list;  

\item Then, the kernel  switches to the context of the task being 
executed. Besides  restoring the working values of the registers
 previously saved, the kernel   correctly sets up  the memory
protection  by updating system registers  \texttt{mpu}$_1$ and \texttt{mpu}$_2$;  
%system registers. 
%\maybe{What ``\texttt{\& \textasciitilde{}7 | 5}'' does here is
%change the rights to only give read and execute permissions.} 

\item Finally the kernel returns from the interrupt, swapping the
values of  \texttt{pc}, \texttt{sp}, and \texttt{flags} with their primed
 counterparts. 
% (actually, \texttt{pc} does not need to be saved). 
Here, the kernel relies on the invariant that \texttt{flags}$'$ is such that
after this instruction, the processor will be in unprivileged mode.

\end{enumerate}

Note that interrupts are masked during  kernel execution,  i.e.  kernel 
execution cannot be interrupted. 

\begin{figure}
  \lstset{language=C,label= ,caption= ,captionpos=b,numbers=none,morekeywords={int32,int64}}
%    numberstyle=\scriptsize,numbers=right, stepnumber=1}
\begin{lstlisting}
 typedef struct { 
   int32 pc, sp, flags;
   int64 code_segment, data_segment;
   struct task *next;
 } Task; 
 Task *cur;
 extern Task task0;  
 
 register int32 sp, pc, flags, sp$'\!$, pc$'\!$, flags$'$;
 register int64 mpu$_1$, mpu$_2$;
 
 void kernel(int32 interrupt_number) {
   /* Interrupt transition, done in hardware. */
   swap(sp$'\!$,$\,$sp); $\,$swap(flags$'\!$,$\,$flags); $\,$swap(pc$'\!$,$\,$pc); $\,$pc$\,$=$\,$&kernel;
   /* Save context unless during boot. */
   if(interrupt_number == RESET)
     { cur = &task0; }
   else
     { cur->sp = sp$'$; cur->pc = pc$'$; cur->flags = flags$'$; }
   /* Scheduler. */
   cur = cur->next; 
   /* Context restore. */ 
   mpu$_1$ = cur->code_segment; mpu$_2$ = cur->data_segment;
   sp$'\!$ = cur->sp; pc$'\!$ = cur->pc; flags$'\!$ = cur->flags;
   /* Return from interrupt (often done in hardware) */
   swap(sp$'\!$, sp); $\,$swap(flags$'\!$, flags); $\,$swap(pc$'\!$, pc); }
\end{lstlisting}
%\end{lstlisting}
\vspace{-3mm}
\caption{Example: a minimalist OS kernel running on \mbox{ideal hardware}}
\label{fig:a-tiny-kernel}
\end{figure}

%\todo{ML:user tasks parameter $\to$ task configuration, mais aussi shape configuration $\to$ shape annotation}
\mysubparagraph{User image} %
System execution depends on the kernel code, but is also \emph{parameterized}
by the \emph{user (tasks) image}, part of the initial state. Here,
the image is made by C code statically allocating  \texttt{Task}
data structures, such that: (1) their \texttt{pc} and \texttt{sp} fields 
point into their respective code and stack sections; (2) the
\texttt{code\_segment} and \texttt{data\_segment} fields should give
the appropriate rights to these sections; (3) \texttt{flags} ensures that 
 execution is not privileged; and (4)  the \texttt{next} fields are such 
that all tasks in the system are in a circular
list. Figure~\ref{fig:exemple-initial-configuration} (Appendix, p.~\pageref{fig:exemple-initial-configuration}) 
gives an example of C code for a user image.% of such parameters. 

% \maybe{As ``\texttt{|6}'' is commented, the \texttt{flags} field of each task
% should be such that the processor will be in unprivileged mode and the MPU is active; if we uncomment this code,
% then there are no requirements on the initial value of \texttt{flags}.}
Getting a working system simply requires to compile the user image  and load it with the
kernel\longversion{ (to communicate they
must agree on a fixed address where \texttt{task0} is located)}. Such static system generation  is common in embedded systems, 
where  OS vendors and application developpers are separate entities. 

% Note that verifying the kernel code independently from the  
% but it makes the system harder to analyze compared to systems where the numnon-parameterized systems

% in embedded
% systems, especially when the hardware only provides a MPU (instead of
% a MMU), like in this
% example. 
%
%\maybe{Paradoxically it makes the system harder
% to analyze than with dynamic task creation, as knowledge about the
% tasks is scattered across two tools: the kernel and the task
% generation tool.}

\subsection{Threats and mitigation}

\mysubparagraph{Attacks} %
This kernel is secure in that it effectively prevents privilege
escalation. However, slight changes in its
behavior could allow a variety of attacks that we illustrate here:

\begin{itemize}

\item {\it Attacks targeting memory safety:} data corruption on one of the
  \texttt{flags} field of a \texttt{Task} would allow a 
   user task to raise its privilege. Corruption can come e.g. from 
  a stack overflow  or following an invalid
  \texttt{next} or \texttt{cur} pointer;  
   
\item {\it Attacks targeting memory protection:}  data corruption on the
  \texttt{code\_segment} or \texttt{data\_segment} field of a
  \texttt{Task} would allow  to extend the memory that the task  can
  access, allowing further code injection or data corruption; 

\item {\it Attacks targeting control-flow:} changing the  control-flow
  of the program to an unexpected execution path can lead to
  privilege escalation. For instance the attacker could use a stack
  smashing attack to jump to the instruction storing to \lstinline{cur->flags}, with a wrong value. 
% or it could just exploit a rare path of the
%  kernel that was never tested, for instance because it relies on a
%  race condition.
\end{itemize}

\noindent {\it Absence  of privilege escalation 
is a fragile property. On the other hand, once verified, it %verifying absence of privilege escalation
% , which rely on having every kernel code path
% being memory-safe and having a working memory protection. This is all
% the more serious that this property is the keystone of the computer
% security: when breached, nothing can be protected against the
% attacker\footnote{Except if protected directly in hardware.}.
implies that all of
the above attacks are impossible or have only a limited impact.}

\mysubparagraph{Guarantee against attacks} %
\label{sec:orgbd358fa}
\label{sec:results-example-state-invariant}
An \emph{(inductive) invariant} is a property that holds for all reachable states,  
because it is {\it initially true}   and {\it inductive}  
(i.e. remains true after
the execution of an instruction). Our method guarantees that no
software-based attack  can lead  to  privilege escalation by
automatically finding an invariant implying  that only the kernel
code can execute with hardware privilege (and that this code cannot be
modified). To be inductive, this invariant also contains properties
about memory, control-flow  and correct working of the
memory protection mechanism. The properties below
 belong to the invariant computed by our method {\it on the~example~kernel}:

\begin{itemize*}
\item \emph{Control-flow safety:} All the instructions executed in
  privileged mode are those of the \texttt{kernel} function, whose
  code is never modified;
\item \emph{Memory safety:} All memory accesses done by the kernel are
  inside its stack, the \texttt{cur} global variable, or inside one of
  the \texttt{Task}. The stack never overflows, and the stack pointer
  at the kernel entry is constant;
\item \emph{Working hardware protection:} The \texttt{flags}$'$
  register  and the \texttt{flags} field in every \texttt{Task}  
   ensure that  execution is unprivileged. The two \texttt{mpu} registers, the \texttt{code\_segment} and
  \texttt{data\_segment} fields in every \texttt{Task}, do not contain write access to the kernel memory;
\item \emph{Shape invariants:} The \texttt{cur} variable (after boot)  and the \texttt{next} field in every \texttt{Task} 
  always point to a \texttt{Task}; each \texttt{Task} is separated
  from the others and from   kernel data.
% \item The kernel execution cannot go wrong (i.e., never triggers an
%   exception)
\end{itemize*}

% \begin{itemize*}
% \item All the instructions executed in privileged mode are those of the
% \texttt{kernel} function, whose code is never modified (control-flow safety);
% \item All memory accesses done by the kernel are inside its stack, the
%   \texttt{cur} global variable, or inside one of the \texttt{Task} (memory safety);
% \item The \texttt{flags}$'$ register, and the \texttt{flags} field in every \texttt{Task}, are each such that the execution is unprivileged;
% \item The \texttt{mpu}$_1$ and \texttt{mpu}$_2$ registers, and the \texttt{code\_segment} and
% \texttt{data\_segment} fields in every \texttt{Task}, never contain write access to
% the kernel code or data;
% \item The \texttt{cur} variable (after the initialization), and the \texttt{next} field in every \texttt{Task},
% always point to a \texttt{Task}; each \texttt{Task} is separated from the others and from the other kernel data (shape abstraction).
% \item The stack never overflows, and the stack pointer at the kernel entry is constant.
% % \item The kernel execution cannot go wrong (i.e., never triggers an
% %   exception)
% \end{itemize*}

\mysubparagraph{The problem of defining Absence of Privilege Escalation} % 
Absence of privilege escalation does not {\it per se} imply memory safety nor control-flow safety\Dash{}even if our approach proves these properties {on the example kernel} 
as a {\it byproduct}.  
Indeed, not all bugs are security critical, and a system suffering from a limited 1-byte stack overflow can still be secure while not respecting strict memory safety.  
Also, control-flow safety and memory-safety are very hard to define at machine code-level, as we lack information about code and data  layout.% \todo{ML: Mais on a dit qu'on y arrivait plus haut.}   

\smallskip 

{\it Hence, a first challenge here is to provide a formal  definition of Privilege Escalation suitable to machine code analysis.}

\subsection{The case for automated verification of OS kernels}

\label{sec:discovering-verifying-inv}
We can distinguish three classes of verification methods: 

% for proving  invariant; all three have been used in previous verification of OS kernels: 

\begin{itemize}
\item \emph{manual}
  \cite{bevier1989kit,richards2010modeling,gu2015deep,xu2016practical,klein2009sel4}: the user has to 
  {\it provide for every program point a candidate invariant}, 
  %, i.e. a property true for states reaching this point; 
  then {\it prove} via a
  proof assistant that every instruction preserves these candidate invariants; 

  %for every instruction $i$ between program
  %points $p_1$ and $p_2$, the outcome of the execution of $i$ from a
  %state satisfying the property at $p_1$ will also satisfy the
  %property at $p_2$; 

\item \emph{semi-automated}\footnote{Some authors call this technique automated; we use the word semi-automated to emphasize the difference with fully-automated methods.} 
  \cite{alkassar2010automated,yang2010safe,dam2013machine,vasudevan2013design,vasudevan2016uberspark,ferraiuolo2017komodo}: the user has to 
  {\it provide} the candidate invariants {\it at some key program 
  points} (kernel entry and exit, loop and function entries and exits) and then 
  use automated provers to verify that all finite paths between these
  points preserve the candidate invariants; 

\item \emph{fully-automated}: a sound static analyzer \cite{cousot1977abstract} {\it automatically
  infers} correct invariants for every program point. The user only {\it provides 
  invariant templates} by selecting or configuring the required
abstract domains.%\todo{ML: check vasudevan2013}
%\todo{ML: vasudevan2013 ne fait pas d'inference d'invariant (bounded model checking), le retirer? Ou le classer dans semi-automated du coup?}
\end{itemize}
%
%
%%%In our toy example alone, all the properties above (except for the last one, but with others) are required for the invariant to be inductive (except in part for the last). % ; moreover this is only a toy
%%% example (with no array, function pointer, complex memory protection
%%% schemes...), 
%%% for which we did not give the complete invariant, which
%%% is even larger when doing the proof at the machine-code level. 
%
%
%
%
Experience has shown that in OS formal verification {\it ``invariant reasoning dominates the proof effort''}\footnote{Klein et al.~reports that 80\% of the effort 
in SeL4 is spent stating and verifying invariants~\cite{klein2009sel4}. }~\cite{walker1980specification,klein2009sel4},   motivating our choice for 
%of computing and verifying invariants using 
fully-automated methods.

\subsection{Challenges}  \label{sec:challenges} 
Besides a suitable definition of Privilege Escalation, 
analy\-zing real microkernels  adds extra
challenges.

%% \mysubparagraph{Scalability} A large number of user tasks raises
%% scalability issues. In our example kernel, the set of addresses where
%% \texttt{cur} can point has a cardinal equal to the number of
%% \texttt{Task}s. If we represent addresses numerically, too many tasks
%% lead to \emph{bad performance} (if we enumerate all addresses) or
%% \emph{imprecision} (if we approximate the information, e.g., using
%% intervals).

%% \mysubparagraph{Summarization} If the number of tasks is not fixed,
%% the kernel cannot find the location of data structures such as
%% \texttt{Task} at fixed addresses in memory. A flat representation of
%% memory~\cite{dam2013machine} is no longer sufficient, and we need more
%% complex representations able to precisely \emph{summarize} memory
%% (i.e. \emph{shape analyses}
%% \cite{sagiv1999parametric,chang2008relational,calcagno2011compositional}). Such
%% analyses often require writing a tedious configuration, which defeats
%% automation.

\mysubparagraph{Machine code analysis}  Binary-level static analysis  is already  very challenging on its own \cite{reps2010there,DBLP:conf/vmcai/BardinHV11,djoudi2016recovering}  
as (1) the control-flow graph is not known in advance because of computed jumps\footnote{Computed jumps are commonly introduced by compilers: {\tt return} operations,  
function pointers, optimized compilation 
of C-like {\tt switch} statements, dynamic method dispatch in OO-languages, etc.} (e.g.~\texttt{jmp @sp}) whose resolution requires runtime values,   (2) memory is a single large array of bytes with no prior typing nor partitioning,  and 
 (3) data manipulations are very low-level (masks, flags, legitimate overflows, low-level comparisons, etc.).

\mysubparagraph{Precondition} Absence of privilege escalation may be true only for
user images that match a given \emph{precondition}. For instance our
example kernel is vulnerable to privilege escalation if initially
\texttt{task0.next} can point inside the kernel stack, or if
\texttt{task0.code\_segment} allows modifying the \texttt{cur} variable.

This suggests verifying the kernel assuming this precondition, and
then checking that the user image satisfy the precondition. But
writing this initial precondition would require a manual effort,  going
against having an automated analysis. 

\mysubparagraph{Boot code} Verifying boot code has its own difficulties:  
  (1)  type invariants  holding  at runtime may not hold
in the initial state, so we have to verify their establishment 
 rather than their preservation; (2)  boot code often  includes hard-to-analyze patterns 
such as dynamic memory allocation  or  creation of
memory protection tables. % for the tasks. 
Consequently, boot code is sometimes left unverified~\cite{klein2009sel4}, and achieving 
perfect automatic analysis on boot code (0 false alarm) is very difficult.%\todo{Je ne suis pas sur que ce sont ces raisons qui ont pousse sel4 a ne pas faire le boot.}

\mysubparagraph{Parametrization} If the number of tasks is not fixed,
the kernel cannot find the memory location of data structures such as
\texttt{Task} at fixed addresses in memory. A flat representation of
memory~\cite{dam2013machine} is no longer sufficient, and we need more
complex representations able to precisely \emph{summarize} memory
(i.e. \emph{shape analyses}
\cite{sagiv1999parametric,chang2008relational,calcagno2011compositional}). Such
analyses often require  a tedious annotations, which defeats
automation.

\mysubparagraph{Concurrency} Operating system kernels are often concurrent 
  because of nested interrupts, preemptive kernel threads, or (as in
  our case study) because the hardware provides several processors. Concurrency brings issues of analysis performance and precision.

\section{Absence of privilege escalation}
\label{sec:formalization-formal-description}

%Formal verification of a system means proving that a formal model of a
%system (the computer running the operating system kernel) satisfies a
%formal property (absence of privilege escalation). 

We present here a
formalization of absence of privilege escalation (APE) suitable to automated verification.  
%of absence of privilege escalation:
Theorem~\ref{th:invariant-implies-noescalation} reduces this
problem so that it can be tackled with standard methods\Dash{}computation
of state invariants.

%
%% \emph{Privilege escalation} is a situation where an untrusted entity 
%% gains access to a \emph{privilege} it should not have, by subverting a
%% security-critical entity entitled to use this privilege.
%%
%% While our main focus is preventing user tasks (or virtual machines)
%% from getting access to kernel (or hypervisor) privileges, this
%% definition and formalization encompasses a variety of situations, like
%% a mobile phone application gaining undue access to the user agenda,
%% virtual machine escape, javascript interpreter escape, etc.  
%
The 
specific instantiation of this formalization to hardware-protected OS kernel can be found in  Appendix~\ref{app:sec:form-hardw-prot}, p.~\pageref{app:sec:form-hardw-prot}. %\ref{app:sec:form-hardw-prot}

% In a nutshell, there is a possible privilege escalation when a state
% can be reached which is both \emph{privileged} and \emph{controlled
%   by an attacker}. The formalization defines these notions of
% privilege and kernel control in general, before instantiating them in
% the case of hardware-protected operating systems.

\subsection{Privilege escalation}

We model a system (comprising the hardware, the operating system and the user tasks) as a transition system
$\langle \mathbb{S}, \mathbb{S}_0, {\to} \rangle$, where $\mathbb{S}$
is the set of all possible states in the system, $\mathbb{S}_0$ the
set of possible initial states, and
${\to} \in \mathbb{S} \times \mathbb{S}$ represents the possible
transitions: ``$s_1 \to s_2$'' means that executing one instruction
starting from state $s_1$ can result in state $s_2$.

\begin{figure}[htbp]

  \small
  \begin{tabular}{|p{.28\columnwidth}|p{.28\columnwidth}|p{.28\columnwidth}|}
    \cline{2-3}
    \multicolumn{1}{c|}{} & \makecell[c]{unprivileged} & \makecell[c]{privileged} \\
    \hline
    \makecell[r]{kernel-controlled}   & \makecell[c]{\tick} & \makecell[c]{\tick} \\
    \hline
    \makecell[r]{attacker-controlled} & \makecell[c]{\tick} & privilege escalation\\
    \hline
  \end{tabular}
  \caption{Privilege escalation happens when the system reaches a
    privileged state controlled by the attacker.}

\end{figure}

A predicate $\textsl{privileged}: \mathbb{S} \to Bool$ tells whether a
state has access to the \emph{privilege} under study. %\todo{formulation bizarre
% non? Proposition: whether a state is privileged}.
In a typical OS kernel, this predicate corresponds to the value of a hardware
register containing the hardware privilege level. The privilege level
restricts  how a state can evolve in a $\to$ transition. For instance on usual
processors, system registers cannot be modified when in an
unprivileged state. 

Two entities are sharing their use of the system, called the
\emph{kernel} and the \emph{attacker}. % (but actually represent all the code that may/may not have access to the privilege). % \todo{ML: precision importante, sinon on se demande pourquoi ``for short'' et pourquoi il y en a 2.}
A predicate
$\textsl{A-controlled}: \mathbb{S} \to Bool$ tells which entity
controls the execution (returning true if it is the attacker, and
false if it is the kernel). In the kernels that we consider, a state
is kernel-controlled if the next instruction that it executes comes
from the kernel executable file and was never modified; all the other
states are considered attacker-controlled.

This formalization corresponds to an attack model where the attacker
chooses the untrusted software running on a system, but cannot change the
security-critical software nor the hardware behavior.

%\todo{This paragraph
%repeats the ``Attack model'' paragraph of the intro. It is also not quite right}

\begin{definition}[Privilege escalation]
We define \emph{privilege escalation} of a transition system
$\langle \mathbb{S}, \mathbb{S}_0, {\to} \rangle$ as reaching a
state which is both privileged and controlled by the attacker.
%\vspace{-3mm}
\[ \mathrm{privilege\ escalation}\quad\triangleq\quad \exists s: \left\{\begin{array}{ll} & \textsl{privileged}(s) \\
                                                                          \land & \textsl{A-controlled}(s) \\
                                                                          \land & \exists s_0 \in \mathbb{S}_0: s_0 \to^{*} s  \end{array}\right.\]
                                                                    where $\to^{*}$ is the transitive closure of the $\to$ relation.
\end{definition}

Thus, an attacker can escalate its privilege by either gaining control
over privileged kernel code (e.g., by code injection), or by leading
the kernel into giving it its privilege (e.g., by corrupting  memory protection tables). 

% Hence, verifying absence of privilege escalation implies that some
% attacks, such as code injection or memory corruption, are impossible;
% or at least limited in scope (since they cannot lead to privilege
% escalation).

\subsection{Parameterized verification}
\label{sec:parameterized-verification}

The previous definition cannot be used for \emph{parameterized
  verification}, because the execution of the system depends on what
is executed by the attacker. We solve this problem by defining a
new semantics for machine code which is independent from the attacker's execution.

\myparagraph{Regular and interrupt transitions}
\label{sec:regular-interrupt-transitions}

We partition the transition relation into \emph{regular transitions} and
\emph{interrupt transitions}: 
\[ s \to s'\quad\triangleq\quad s \overset{\scriptscriptstyle \textsl{interrupt}}{\to} s'\ \lor\ s \overset{\scriptscriptstyle \textsl{regular}}{\to} s' \]

The \emph{regular transition} $\overset{\scriptscriptstyle \textsl{regular}}{\to}$
from states $s$ to $s'$ corresponds to the execution of an instruction
$i \in \mathbb{I}$, the set of all instructions. 
\[ s \overset{\scriptscriptstyle \textsl{regular}}{\to} s'\quad\triangleq\quad
  s' \in \textsl{exec}(s,\textsl{next}(s)) \]

%\todo{ML: On devrait decrire les instructions ici, en simplifié par rapport aux dbas. load,store,move, y compris pc, et c'est tout.}
$\textsl{next}: \mathbb{S} \to \mathbb{I}$ fetches and decodes the
next instruction, while
$\textsl{exec}: \mathbb{S} \times \mathbb{I} \to
\mathcal{P}(\mathbb{S})$ executes this instruction (which may be
non-deterministic). We assume that regular transitions 
$s \overset{\scriptscriptstyle \textsl{regular}}{\to} s'$ either preserve the current
privilege level or evolve from a privileged state to an unprivileged
state, but {\it cannot} evolve from an unprivileged state to a
privileged one. The \emph{interrupt transition} $\overset{\scriptscriptstyle \textsl{interrupt}}{\to}$ is
thus the only way to evolve from  unprivileged
 to  privileged. In OS kernels, it 
corresponds to the reception of hardware or software interrupts. 
% which changes the privilege level, changes the program counter to a
% specific address (the \emph{kernel entry point}), and performs other
% operations such as saving the values of some registers.

\myparagraph{Empowering the attacker}
\label{sec:enpowering-attacker}

%\xxx{parametric or parameterized? parameterized seems more common}
Note that the $\to$ transition is 
defined only when  the instruction under execution  is known, which
prevents parameterized verification. That is why we define a new transition
system
$\langle \mathbb{S}, \mathbb{S}_0, \multipleinstruction \rangle$ with
the same sets $\mathbb{S}$ and $\mathbb{S}_0$, but with a new  transition
relation $\multipleinstruction$. 

We first define the $\multipleinstructionattacker$ relation  which
over-approximates the transitions that an attacker can effectively
perform (i.e., we make the attacker more powerful). Instead of being
able to execute only one known instruction $next(s)$, the attacker will now
be able to execute sequences of arbitrary instructions:
\[s \multipleinstructionattacker s''\quad\triangleq\quad{}s'' = s\;\lor\;(\exists i, s': s \multipleinstructionattacker s'\,\land\,s'' \in \textsl{exec}(s',i))\]

The $\multipleinstruction$ relation restricts the ability to execute arbitrary
instructions to attacker-controlled states; when a state is kernel-controlled, the
normal transition apply. In addition, interrupt transitions are also
possible in attacker-controlled states.
% The $\multipleinstruction$ relation still restricts the ability of the attacker to execute arbitrary
% instructions, as privilege restrictions still apply. 
% %when a state is kernel-controlled, the
% %normal transition  applies again. 
% In addition, interrupt transitions are 
% possible in attacker-controlled states. 
\[ s \multipleinstruction s' \quad\triangleq\quad s \to s'\ \ \lor\ \  (\textsl{A-controlled}(s) \land\ s \multipleinstructionattacker s') \]

\begin{theorem}
  The set of reachable states for the
  $\langle \mathbb{S}, \mathbb{S}_0, {\to} \rangle$ transition system
  is included in the set of states reachable for
  $\langle \mathbb{S}, \mathbb{S}_0, {\multipleinstruction}
  \rangle$.
\end{theorem}
\begin{proof}
  This follows directly from the fact that for every $s$,
  $\{s': s \to s' \} \subseteq \{ s': s \multipleinstruction s' \}$
\end{proof}

\begin{corollary} \label{th:privilege-escalation-new-transition-system}
  If there are no privilege escalation in the transition system
  $\langle \mathbb{S}, \mathbb{S}_0, {\multipleinstruction}\rangle$,
  there are also no privilege escalation in the transition system
  $\langle \mathbb{S}, \mathbb{S}_0, {\to}\rangle$
\end{corollary}
\begin{proof}  
  This is the contrapositive of the fact that if a state exists in
  $\langle \mathbb{S}, \mathbb{S}_0, {\to} \rangle$ where privilege
  escalation happens, this state also exists in
  $\langle \mathbb{S}, \mathbb{S}_0, {\multipleinstruction}\rangle$.
\end{proof}

Thanks to
Corollary~\ref{th:privilege-escalation-new-transition-system}, the
$\langle \mathbb{S}, \mathbb{S}_0, {\multipleinstruction}\rangle$
transition system can be used to prove  absence of privilege
escalation instead of
$\langle \mathbb{S}, \mathbb{S}_0, {\to}\rangle$,  with the benefit of
establishing this proof independently from a particular concrete attacker. % (works for any attacker). 

\subsection{Proof strategy}

We now  show  that we can recast absence of privilege
escalation to ease its formal verification.

\myparagraph{Absence of privilege escalation as a state property}

A \emph{state property} $(\in \mathbb{S} \to Bool)$ is a predicate
over states.
%\maybe{A state property is isomorphic to the set of states  that satisfy $p$ (the extension of $p$)
% \maybe{A state property $p$ can be viewed not only as a
% predicate but also as a set, namely the set
% $\{ s \in \mathbb{S} \mid p(s) \}$.}
A state property is \emph{satisfied} if it is true in every reachable state.

\begin{theorem}
A transition system does not have privilege escalation if, and only if, it satisfies the $\textsl{secure}$ property, where%
\[ \textsl{secure}(s) \triangleq \lnot(\textsl{A-controlled}(s) \land
  \textsl{privileged}(s)) \]
\end{theorem}

% \textcolor{red}{XXX: du coup, ca peut etre prouvé par n'importe quel model checker qui verifie des state properties. }
% \textcolor{red}{XXX: the idea is to compute and prove an \emph{invariant}, i.e. a sound over-approximate of the set of reachable states, and check that not states in this invariant can be both attacker-controlled and privileged}

Recasting absence of privilege escalation as a state property provides
a method to prove this property. The idea is to find a \emph{state
  invariant}: a property $p$ which holds on every initial state
$s \in \mathbb{S}_0$  and is inductive (i.e.
$p(s) \land s \to s'\ \Rightarrow\ p(s')$), and thus holds on each
reachable state. If this invariant $p$ is \emph{stronger} than
\textsl{secure}, this proves that we cannot reach a state which is both
attacker-controlled and privileged.

%(i.e.
%$\forall s \in \mathbb{S}: p(s) \Rightarrow \textsl{right}(s))$)

% XXX: est-ce qu'on a besoin de parler d'over-approximation, ou est-ce qu'on peut juste dire qu'on peut calculer des invariants automatiquements? Ou est-ce que, au contraire, je ne devrais parler que d'ensembles et jamais de formule?

% sound over-approximation of the set of 

% by a static analyzer
% based on abstract interpretation \cite{cousot1977abstract}. Indeed
% abstract interpretation is a sound method that can compute and verify
% invariants (including state invariants). The idea is to compute a
% sound over-approximation of the set of reachable states, and check
% that no states in this over-approximation can be both
% attacker-controlled and privileged.

\myparagraph{Reasoning about consequences of privilege escalation}

Given how we empowered the attacker, we can also reason about the
\emph{consequences} of a privilege escalation: indeed the definition
of $\multipleinstruction$ implies that the attacker will do everything
that it can do. Thus, to prove absence of privilege escalation in 
$\langle \mathbb{S}, \mathbb{S}_0, {\to}\rangle$, 
one can prove that a {\it bad  consequence} of privilege
escalation never happens in 
$\langle \mathbb{S}, \mathbb{S}_0, {\multipleinstruction}\rangle$.

In the rest of the paper, we will consider that {\it bad consequence} is
given by the following assumption:

\begin{assumption}
  Running an arbitrary sequence of privileged instructions allows to
  reach any possible state.
\end{assumption}

This assumption is reasonable for every system that confines an
adversarial code in some kind of container, 
%
%: virtual machines in a
%hypervisor, processes in an operating system, bytecode executed by an
%interpreter, code supervised by a runtime monitor, etc. 
%
as escaping such a
container means getting rid of any restrictions on the possible
actions. 
In the case of hardware protected OSes,
executing an arbitrary sequence of instructions with hardware privilege
allows changing any register or memory location, thus  reaching any
state. Using this assumption we can prove the following theorem:

\begin{theorem} 
  If a transition system
  $\langle \mathbb{S}, \mathbb{S}_0, {\multipleinstruction} \rangle$
  is vulnerable to privilege escalation, then the only satisfiable state property
  in the system is the \emph{trivial state invariant $\top$}, true for every state.
  \label{th:priv-escal-state-invariant}
\end{theorem}
\begin{proof}
  If a privilege escalation vulnerability exists, then any state can be
  reached, as executing an arbitrary sequence of instructions starting
  from a privileged state can lead to any state. The only state
  invariant that is true on every state is~$\top$.
\end{proof}

We define as \emph{non-trivial} any state invariant different from~$\top$.

\longversion{Note that on the contrary, if any state can be reached,
  then obviously a privileged attacker-controlled state can be
  reached, provided that such a state exists (if it does not, then
  verifying  absence of privilege escalation is pointless).}

\begin{theorem}
  If a transition system satisfies a non-trivial state invariant, then it is
  invulnerable to privilege escalation attacks.%
  \label{th:invariant-implies-noescalation}
\end{theorem}
\begin{proof} By contraposition, and the fact that state invariants
  are satisfiable state properties.
\end{proof}

Theorems \ref{th:priv-escal-state-invariant}  \& \ref{th:invariant-implies-noescalation}  have two crucial practical implications:
\begin{itemize}
\item If privilege escalation is possible, \emph{the only state property that holds in the system is $\top$}, making it impossible to prove definitively
  any useful property. Thus, \emph{proving  absence of
    privilege escalation is a necessary first step} for any formal
  verification of an OS kernel; 

\item The proof of \emph{any} state property different from $\top$ implies as a
  byproduct the existence of a piece of code able to protect itself
  from the attacker, i.e. a kernel with protected privileges. In
  particular, we can prove absence of privilege escalation automatically, by successfully finding \emph{any} non-trivial state
  invariant with a sound static analyzer.
\end{itemize}

\section{\binseccodex\ for OS verification}
\label{sec:general-method}

In Section~\ref{sec:formalization-formal-description} we have
reduced the problem of proving absence of privilege escalation   to
the problem of finding a non-trivial system invariant. We now detail our methodology to find such an invariant. 

\subsection{Background: general principles} 

\mysubparagraphbis{Abstract interpretation}~\cite{cousot1977abstract} is a method for building
sound static analyzers that can {\it infer program invariants} and {\it verify program properties}. 
Using it, we can compute a
set of states guaranteed to be {\it larger} than the set of reachable states for the
$\langle \mathbb{S}, \mathbb{S}_0, {\multipleinstruction} \rangle$
transition system. If this computed set is different from
$\mathbb{S}$, this proves that a {\it non-trivial invariant exists}. 

% Abstract interpretation works by combining \emph{abstract domains},
% i.e. sets of \emph{abstract values} that are finite representations of a set. Formally,
% the \emph{meaning} of an
% abstract domain $\mathbb{D}^\sharp$ containing elements $d^\sharp$ representing a set of elements in
% $\mathbb{D}$, is given by its \emph{concretization}, which is a
% function $\gamma_\mathbb{D}^\sharp: \mathbb{D}^\sharp \to
% \mathcal{P}(\mathbb{D})$. For instance, intervals are finite
% representations of (possibly infinite) sets of integers using two bounds:
% $\mathbb{D}^\sharp = \left(\mathbb{Z} \cup \{-\infty,+\infty\}\right)^2$ and
% e.g., $\gamma\left(\langle 3,+\infty\rangle\right) = \{ x \in \mathbb{Z} \mid 3 \le x \}$.

Abstract interpretation works by computing \emph{abstract values}\Ldash{}elements of \emph{abstract domains}\Rdash{}representing a set of states.  Abstract domains are
iteratively computed until a \mbox{(post-)}fixpoint is reached, yielding an
over-approximation of the set of reachable states. This
over-approximation allows to trade precision for termination. %
%
% The basic idea behind abstract interpretation is to compute over finite representations of families of {\it abstract} sets of states 
% ({\it abstract domains})  instead of concrete states in order to  trade  
% precision for termination. The abstract set of reachable states is then computed in an iterative fixpoint manner.  
%
% and manage to effectively compute an overapproximation of the set of reachable states 
% -- in an iterative fixpoint computation manner. 
%
Such abstract domains range  from simple (intervals) to complex (polyhedra), offering various trade-offs 
between precision and scalability, and domains can be combined together~\cite{cousot1979systematic}. 

In practice, designing an abstract interpreter amounts to defining or choosing a combination of abstract domains tailored to the problem at hand.%, and tuning several heuristic choices regarding fixpoint iteration and domain combination.  
%\todo{A merger avec appendix}

\mysubparagraph{Intermediate representation for machine code analysis}   
Intermediate Representations (IR) \cite{Brumley2011,bardin2011bincoa,kaist-bar-workshop} are the cornerstone of modern binary-level code  
analyzers, used to lift  the different binary Instruction Set
Architectures  into 
a single and simple language. 
We rely on the IR of \binsec~\cite{bardin2011bincoa,Djoudi2015}, called DBA\Dash{}other
IRs are similar. Its syntax is shown Figure~\ref{fig:dba_ir}.

% It is Explicit \cite{DBLP-conf/kbse/KimFJJOLC17} and Self-contained, meaning an
% instruction can only update a single variable in the context and that there is
% no call to built-in functions.

{

\floatstyle{boxed}
\restylefloat{figure}

\newcommand{\te}[1]{{\sf{#1}}}
\begin{figure}[htbp]
  \renewcommand\te[1]{$\langle$#1$\rangle$}
  \centering\scriptsize
  \begin{tabular}{rcl}
    \te{stmt} & {\sf :=} & {\sf store} \te{e} \te{e} $\mid$ \te{reg} $\bm\leftarrow$ \te{e} $\mid$ {\sf goto} \te{e} \\
              &  &   $\mid$ \te{stmt};\te{stmt} $\mid$ {\sf if} \te{e} \te{stmt} {\sf else} \te{stmt}\\
    \te{e} & {\sf :=} & \te{cst} $\mid$ \te{reg} $\mid$ {\sf load} \te{e}
                      $\mid$ \te{unop} \te{e} $\mid$  \te{e} \te{binop} \te{e} \\
    \te{unop} & {\sf :=} & ${\bm \neg}$ $\mid$ ${\bm -}$ $\mid$ {\sf uext$_{n}$}
                          $\mid$ {\sf sext$_{n}$}    $\mid$ {\sf extract$_{i .. j}$} \\
    \te{binop} & {\sf :=} & \te{arith} $\mid$ \te{bitwise} $\mid$ \te{cmp}
                           $\mid$ {\sf concat} \\
    \te{arith} & {\sf :=} & ${\bm +}$ $\mid$ ${\bm -}$ $\mid$ ${\bm \times}$
                           $\mid$ {\sf udiv}    $\mid$ {\sf urem}
                           $\mid$ {\sf sdiv}   $\mid$ {\sf srem} \\
%    \te{bitwise} & {\sf :=} & ${\bm \land}$ $\mid$ ${\bm \lor}$ $\mid$ ${\bm \oplus}$
%                              $\mid$ $\ll$ $\mid$ $\gg_s$ $\mid$ $\gg_u$ \\
    \te{bitwise} & {\sf :=} & {\sf and} $\mid$ {\sf or} $\mid$ {\sf xor}
                             $\mid$ {\sf shl} $\mid$ {\sf shr} $\mid$ {\sf sar} \\
    \te{cmp} & {\sf :=} & ${\bm =}$ $\mid$ ${\bm \neq}$
                         $\mid$ $>_{u}$   $\mid$ $<_{u}$
                         $\mid$ $>_{s}$         $\mid$ $<_{s}$
  \end{tabular}
  \caption{Low-level IR for binary code}
  \label{fig:dba_ir}
\end{figure}

}

\noindent DBA is a minimalist language  comprising only a single type of
elements (bitvector values) and four types of instructions: register assignments, memory stores, conditionals, and jumps (static and computed). Expressions contain memory load, as well as standard modular arithmetic and bit-level operators.  
Values are stored in  registers or in memory (a single byte-level array), the imperative semantics is standard.  
This is  enough to encode the
functional semantics of major instructions sets\Dash{}including x86 and ARM.
%\footnote{Actually, a single machine instruction is translated into a local control-flow graph whose nodes are DBA instructions, called a DBA block.}

\subsection{Key ingredients} \label{sec:key}

In Section~\ref{sec:formalization-formal-description} we have
simplified the problem of proving absence of privilege escalation, to
the problem of finding a non-trivial system invariant. We intend to
use binary-level static analysis to
infer and verify this invariant automatically.

%We intend to develop a static analysis tailored to binary-level APE verification 
%of microkernels. 
%
%Thanks to Theorem \ref{th:invariant-implies-noescalation} we know that  if we 
%find a non trivial invariant of the code under analysis  then this code 
%satisfies APE. 
%
In practice, this requires developing a static analysis for machine
code  and run it
on the system loop (Figure~\ref{fig:system-control-flow}). We  
define a special transfer function to handle the
$\multipleinstructionattacker$ (user code) transition, that removes
knowledge about the contents of registers and memory
\textsl{accessible} from the loaded memory protection table.
This approach should work if the system is not parameterized,
i.e.~when the system is completely known, or at least the memory layout of
user tasks is known \cite{bevier1989kit,dam2013machine}. 

Yet, we face two problems here: 
(1) binary-level static analysis  is already  very challenging on its own \cite{reps2010there,DBLP:conf/vmcai/BardinHV11,djoudi2016recovering} (Section \ref{sec:challenges})%
%as  control-flow graph is not known in advance,  memory is a single large array of bytes with no prior typing nor partitioning and 
%data manipulations are very low-level (masks, flags, legitimate overflows, etc.) 
\Dash{}see, e.g., the boot code analysis in our case study; %, analyzing the boot code was very difficult;    
(2) the  systems we are interested in {\it are} parameterized.

\smallskip

{\it Our methodology  builds upon three key ingredients.}

\mysubparagraph{Key 0:  An up-to-date binary-level static analysis} We build a 
state-of-the-art {\it sound} static analyzer for machine code (Section~\ref{sec:static-analysis}), 
picking among the best practices from the literature\cite{reps2010there,DBLP:conf/vmcai/BardinHV11,djoudi2016recovering,kinder2009abstract,sepp2011precise,BrauerKK10}.   
Interestingly, while prior works is partitioned\cite{reps2010there} into {\it raw binary} analysis  
(no assumption, adequate to adversarial analysis such as malware but extremely difficult to get precise) 
and {\it standardly-compiled code} analysis 
(with {\it hard-coded} extra assumptions\footnote{Typically, assumptions  on control flow (trust in an external disassembler)
  or memory partitioning (e.g., stack is separated from the rest of memory).}),
 our own method does target standardly-compiled code but the extra assumptions are {\it explicit}  (shape annotations) and  {\it fully checked}.   
%by the analysis. 
%\todo{ML: attention, reps dit explicitement que ses hypothèses sont vérifiées. SB: citer celui de 2004 ? ML: En fait ce papier dit
%  qu'ils peuvent etre unsound si ils emettent des alarmes}

\mysubparagraph{Key 1:  A type-based weak shape abstract domain} We propose a
novel weak shape abstract domain
(Section~\ref{sec:shape-abstract-domain}) based on the layout of types
in memory. This fulfills many roles:

%  it is such that  (in our implementation, we
% extracted them from debug information in the kernel executable). This
% type system is not that of C (which is too weak to guarantee type
% safety, and too strong to properly handle access to fields using
% numerical offsets), but can be derived from C types augmented with
% simple \maybe{SAL-like?} annotations.

\begin{itemize}

\item Being a shape domain, it can be used to \emph{summarize} the
  memory and allows representation of addresses which is
  \emph{scalable} (no enumeration) and \emph{precise}; 

\item Types and type invariants can encode the \emph{precondition} on
  (the shape of) user images in a simple way;

\item Being based on types, \emph{most of the annotations can be
    extracted from the type declarations} of the kernel code. This
  provides a base set of annotations for the shape domain, that is
  easy to strengthen using type invariants.

\end{itemize}

\mysubparagraph{Key 2: Differentiated handling of boot and runtime} %
%
%Our weak shape abstract domain is tailored to verifying the
%preservation of data structures; this is the key to kernel
%verification~\cite{walker1980specification,alkassar2010automated} and
%will work well on the kernel runtime. Using this shape domain allows a
%parameterized analysis, reusable and efficient since details
%of the user tasks are abstracted.
%
Contrary to runtime code, the boot code does not preserve data structure invariants but
establishes them. Thus, parameterized verification (with 0 false alarms) of boot code is complex.
%(e.g. when memory protection tables are created dynamically).
On the other hand, when the user image is
known, boot code execution is almost deterministic (except from
small sources of non-determinism: multicore handling, device
initialization, etc.), mostly analyzable by simple
interpretation.

We propose an asymmetric treatment of boot code and runtime: our
parameterized analysis completely ignores the alarms in boot code,
meaning that when reaching \emph{0 alarm in the runtime}, we get a
\emph{system invariant} $I$ \emph{under the precondition} that the
state after boot is in $I$. The latter is then verified by
analyzing (mostly, interpreting) the boot code with a given user tasks
image.

%
% Thus, we just have to concentrate on getting an invariant for the
% runtime. Then, to verify the boot code, we just have to check that the
% states resulting from the boot match this invariant (reusing inclusion
% checking, a core operation of the static analysis).

\subsection{Putting things together: the 3-step methodology}

\begin{figure}[tbp]
  \small
  \begin{tikzpicture}[xscale=2.66,
    every node/.style={align=center},
    main node3/.style={draw=green!50!black!50,top color=white,bottom color=green!50!black!20},
    main node2/.style={draw=orange!50!green!50!black!50,top color=white,bottom color=orange!50!green!50!black!20},
    main node1/.style={draw=orange!80!green!50!black!50,top color=white,bottom color=orange!80!green!50!black!20}    
    ]

    \draw[black!20,thin] (0.53,2.7) -- (0.53,-1.7);    
    \draw[black!20,thin] (1.45,2.7) -- (1.45,-1.7);
    \draw[black!20,thin] (2.55,2.7) -- (2.55,-1.7);

    \draw[black!20,thin] (0,0.88) -- (3.4,0.88);
    \draw[black!20,thin] (0,-0.6) -- (3.4,-0.6);        
    
    \begin{scope}[shift={(0,0)}]
      \node[above left] at (0.54,0.45) {\textsc{Method}};
      \node[main node1,draw,ellipse,inner sep=1pt]   (extraction) at (1,0) {Automated\\ annotation};
      \node[main node2,draw,ellipse,inner sep=1pt]   (statican) at (2,0) {Static\\ analysis};
      \node[main node3,draw,ellipse,inner sep=1pt] (usercheck) at (3,0) {User tasks\\checking};            
    \end{scope}

    \begin{scope}[shift={(0,2.8)}]
    \node[below left] at (0.54,0) {\textsc{Input}};
    \node[rectangle,below,align=left] at (1,0) {\textbf{Kernel}:\\- types\\(optionally \\  manually \\strengthened)};
    \node[rectangle,below,align=left] at (2,0) {\textbf{Kernel}:\\- initial state\\- boot code\\- runtime code};%\\- shape configuration}; 
    \node[rectangle,below,align=left] at (3,0) {\textbf{Kernel}:\\- initial state\\- boot code\\[3pt]\textbf{User tasks}:\\- initial state};
    
%    \node[draw,rectangle,above] at (3,0) {User tasks\\initial state\\Kernel\\initial state\\boot code};    
    \end{scope}

    \begin{scope}[shift={(0,-0.6)}]
    \node[below left] at (0.54,0) {\textsc{Output}};
    \node[rectangle,below] at (1,0) {Shape annotations\\configuring\\the analysis};
    \node[rectangle,below] at (2,0) {System invariant \\ under  precondition \\ on the state after boot};
    \node[rectangle,below] at (3,0) {Invariant holds};    
    \end{scope}

    \draw[thick,->] (extraction) edge[loop above] (extraction);

    \draw[very thick,->] (extraction) edge (statican);
    \draw[very thick,->] (statican) edge (usercheck);
    
  \end{tikzpicture}

  \caption{3-step methodology for system invariant computation.}
  \label{fig:3-step-methodology}
  
\end{figure}

%Regarding the specific challenges of binary-level OS verification (Section \ref{sec:challenges}) and our 
%key insights (Section \ref{sec:key}),  

Our methodology (Figure~\ref{fig:3-step-methodology}) consists in three steps:

%\todo{ML: shape configuration -> analysis configuration?}
\mysubparagraph{1. Automated annotation} We automatically extract 
type declarations from the kernel (currently using the DWARF debug section in either the kernel or user image executable,
but extraction from source code is also feasible), 
and (optionally) strengthen them with simple type invariant annotations 
 (presented in Appendix~\ref{app:sec:annotation-language}). This produces {\it shape annotations} configuring the kernel analysis;  

% These annotations
% should be describing the memory layout in the kernel runtime, ignoring
% the fact that some of them may not hold in the boot code.

\mysubparagraph{2. Parameterized static analysis} Then, we launch the
static analyzer on the kernel (boot code and runtime code) using the
shape annotations configuring the shape abstract domain. The result is an
invariant for the runtime $I$ under the precondition that the state
after boot satisfies $I$.

\mysubparagraph{3. User tasks checking} Given a user image, 
we then interpret the boot code and check that its final state indeed satisfies 
the invariant precondition on the kernel runtime. If it does, we have a system
invariant\Dash{}if not trivial it guarantees absence of privilege escalation.

%\subsubsection*{Remarks}
%\medskip 
%Note that even if we check the boot code separately, the
%static analysis still has to run through the boot code to find the
%current invariant (e.g., in Figure~\ref{fig:a-tiny-kernel}, analyzing
%the boot code is necessary to ensure that \texttt{cur} is initialized
%in the runtime). 
%
%
%
%% Note also that the analysis of the boot code is done
%% assuming that the types invariant of the runtime already hold before
%% boot, which might not be true. This works because in the
%% initial state, the parts of the data that do not match the type
%% invariants will generally be overwritten by the boot code without
%% being read. Doing this avoids the effort of writing down the expected
%% format of the user tasks parameters before boot. 
%

%\todo{ML: Type invariants is only template invariant: verifier qu'il est qq part, pour contrer l'argument ``on donne deja l'invariant''}
%Note that while we provide a template
%invariant under the form of \emph{type} invariants, the real invariant
%about the kernel \emph{values} is computed automatically by our analysis.

\subsection{Binary-level static analysis}
\label{sec:static-analysis}
\label{sec:static-analysis-machine-code}

%Sound static analysis of machine code is known to be challenging~\cite{reps2010there}. 
Outside of our type-based weak shape abstract domain (Section~\ref{sec:shape-abstract-domain}),   
%used to represent the part of the memory holding the user tasks parameters,
our analyzer builds on state-of-the-art techniques for machine
code analysis: 

%\begin{itemize}
%\item \emph{Control flow.} 

\mysubparagraph{Control flow} Control and data are strongly interwoven at binary level, since  
%  For instance function returns or switch statements are translated to
  resolving computed jumps (e.g., \texttt{jmp @sp})  requires to precisely track  
  values in registers and memory,  
 %Hence, Control Flow  needs to be recovered 
 % simultaneously with Data Flow~\cite{kinder2009abstract,DBLP:conf/vmcai/BardinHV11}, 
  while on higher-level languages  
  retrieving a good approximation of the CFG is simple. Following Kinder or V\'edrine et al.~\cite{kinder2009abstract,DBLP:conf/vmcai/BardinHV11}, 
  our analysis computes simultaneously the CFG and value abstractions;       

%: It has to retrieve the control flow graph
%  out of machine code.

%\item \emph{Value.} %The most important issue when designing abstract
  %interpreters is the balance between analysis performance and
  %precision. 

\mysubparagraph{Values}  We mainly use efficient non-relational abstract domains
  (reduced product of the signed and unsigned meaning of
  bitvectors~\cite{djoudi2016recovering}, and congruence
  information~\cite{granger1989congruences}), complemented   with
   symbolic relational information
  \cite{mine2006symbolic,gange2016uninterpreted,djoudi2016recovering} for local simplifications 
  of sequences of machine code;
  %\todo{@XR, ON: une reference sur le variable bound chehcking comme on fait?}

%\item \emph{Memory.} 

\mysubparagraph{Memory} Our memory model is ultimately byte-level in order to deal with 
  very low-level coding aspects of microkernels. Yet, 
  as representing each memory byte  separately is inefficient and
  imprecise,  we use a stratified representation of memory caching 
  multi-byte loads and stores, like Min\'e~\cite{mine2006field}.  In
  addition, % we do not map all the addresses in memory upfront, but we
  the tracked memory addresses are found on demand;  
  %  and use only those that are accessed by the kernel during the analysis; 

%\item \emph{Extra precision.} 

\mysubparagraph{Precision} To
  have sufficient precision, notably to enable strong updates to the
  stack, our analysis is flow and fully context-sensitive (i.e. inline
  the analysis of functions), which is made possible by the small size
  and absence of recursion  typical of microkernels. Moreover we unroll
  the loops when the analysis finds a bound on their number of iterations;
  
  % In machine code, this requires precise heuristics to identify when
  % the control changes of functions.

%\item \emph{Concurrency.} 
  %Operating system kernels can be concurrent 
  %because of nested interrupts, preemptive kernel threads, or (as in
  %our case study) because the hardware provides several processors. 

\mysubparagraph{Concurrency}  We handle  shared memory zones 
  through   a flow-insensitive abstraction  making them
  independent from  thread
  interleaving~\cite{venet2004precise}. Our weak shape abstract domain
  (Section~\ref{sec:type-abstract-domain}) represents one part of the
  memory in a flow-insensitive way. For the other zones we identify
  the shared memory zones by intersecting  the addresses read and
  written by each thread \cite{mine2011static}, and only perform weak
  updates on them.
%\end{itemize}

%Our static analysis is made by composition~\cite{cousot1979systematic} of abstract domains. 

\smallskip 

\noindent More formal details are given in Appendix~\ref{app:sec:formal-description-sound-static-analysis} (p.~\pageref{app:sec:formal-description-sound-static-analysis}).
% except for our novel type-based weak shape abstract domain, described in Section~\ref{sec:shape-abstract-domain}.  

%\section{Static analysis of parameterized memory}
\section{A type-based weak shape abstract domain}
\label{sec:shape-abstract-domain}
\label{sec:type-abstract-domain}

% Following our method, we require an abstract domain to verify the
% preservation of memory invariants (like the absence of overlapping
% between several objects of type \texttt{Task}) in the
% parameterized part of memory, whose layout depends on the application
% being executed.  Abstract domains sufficiently precise to fulfill this
% role are called shape domains.

% We thus need a shape domain, which can \emph{handle unannotated
%   machine code}, is \emph{easy to setup} (shape analyses often require
% tedious parametrization by a formal method expert), and
% \emph{efficient} (shape analyses can be slow even on small codes).

The weak shape abstract
domain is designed to  verify the preservation of memory layouts
expressed using a particular dependent type system,  tailored to the
analysis of machine code (e.g., using subtyping to access
field offsets implicitly). 

%\todo{ML:based on type easy to setup -> verifier que c'est dit ailleurs qu'en intro. weakness efficient and suitable to concurrent verification -> aussi}
%Being based on types make it easy to set up (types are easily
%understood by non-experts, and most of the type information can be
%retrieved automatically); the weakness make it efficient and suitable
%to the verification of concurrent programs (it does not require
%tracking the contents of memory).

\subsection{Types and labeling of memories}
\label{sec:typed-memories}

Our memory abstraction considers only memories that follow constraints
coming from how C compilers lay out values in memory.  We interpret a
C declaration like \texttt{Task task0} as having two distinct
meanings: first, that the memory addresses occupied by \texttt{task0}
is \emph{labeled} as being a \texttt{Task}; second, that the contents
of \texttt{task0} are in the set of values that a \texttt{Task} should
contain. In particular, \texttt{task0.next} should only contain
addresses with a \texttt{Task} label.

Formally, we will define \emph{types}
$\mathbb{T}$, \emph{labelings}
$\mathscr{L} \in \mathbb{A} \to \mathbb{T}$ of addresses,
\emph{interpretations}
$\llparenthesis \cdot \rrparenthesis_{\mathscr{L}}$ of types as set of
values, and memories $m \in \mathbb{A} \to \mathbb{V}$; and we will consider only
$(m,\mathscr{L})$ pairs such that $\mathscr{L}$ \emph{is a labeling} for $m$, a relation defined as
\[ \forall a \in \mathbb{A}:\quad m[a] \in \llparenthesis \mathscr{L}(a) \rrparenthesis_{\mathscr{L}} \]
and meaning that each address $a$ in a memory $m$ should contain a value
$m[a]$ whose type is $\mathscr{L}(a)$.

\mysubparagraph{Types} We build a new type system upon that of C, with the following
syntax for types:
\[ \mathbb{T} \ni t ::= \mathtt{word} \mid \mathtt{int} \mid s_k \mid t* \]
where, informally, $\mathtt{word}$ represents any value, $\mathtt{int}$
represents any value used as an integer (i.e., not as a pointer), $s_k$
represents the $k$th byte in a C structure $s$, and $t*$ represents a
pointer to $t$. Note that $s*$ in C (a pointer to a structure $s$) is
 translated in our system into $s_0*$ (a pointer to the
$0$th byte of $s$).

\mysubparagraph{Labeling and memory} %
Consider the C code of Figure~\ref{fig:typing-example}(a)\footnote{To
  simplify the presentation, we assume that every word,
  including pointers, has size one byte.}. Compiling and running it
produces (Figure~\ref{fig:typing-example}(c)) both a \emph{memory}
$m \in \mathbb{A} \to \mathbb{V}$ (mapping from addresses to values)
and a \emph{labeling} $\mathscr{L} \in \mathbb{A} \to \mathbb{T}$
(mapping from addresses to types). The address $\mathtt{0x01}$
corresponds to \texttt{\&u.data}, which should hold the byte at offset 1 in a value of type
\texttt{foo}; thus the label of $\mathtt{0x01}$, $\mathscr{L}(\mathtt{0x01})$,
is $\mathtt{foo_1}$.

\mysubparagraph{Subtyping} %
Following the C types definition, we know that addresses at offset 1
in \texttt{foo} structures are used to store pointers to a
\texttt{bar} structure (i.e. addresses whose label is
$\mathtt{bar_0}$). Thus, we want to say that $\mathtt{0x01}$ is also
labeled by $\mathtt{(bar_0)*}$.

We express this property as a \emph{subtyping} relationship:
$\mathtt{foo_1}$ is a \emph{subtype} of $\mathtt{(bar_0)*}$ (written
$\mathtt{foo_1} \sqsubseteq \mathtt{(bar_0)*}$).  The inclusion $t \sqsubseteq u$
means that addresses labeled by $t$ are also labeled by $u$.
% the set of cells
% holding byte 1 in \texttt{foo} structures is a subset of the cells
% holding pointers to \texttt{bar} (i.e., pointers to \texttt{bar} can
% exist outside of \texttt{foo} structures).
This subtyping relationship is very important for programs written in
machine languages, and deals in particular with the following
issue. While accessing the field of a structure (e.g., \texttt{p =
  \&(p->next)}) is an explicit operation in C, it is implicit in
machine code (the previous statement is a no-op when \texttt{next} has
offset 0), thus requiring pointers to a structure to simultaneously
point to its first field.

Figure~\ref{fig:typing-example}(b) provides the subtyping relations
derived from Figure~\ref{fig:typing-example}(a). Formally, %the subtyping
% relationship is defined as follows:
% \begin{itemize}
% \item $t \sqsubseteq t' \quad\Longrightarrow\quad t* \sqsubseteq t'*$
if $s$ contains a $t$ at offset $o$, then for all $k$ such that
$0 \le k < \operatorname{sizeof}(t)$: $s_{o+k} \sqsubseteq t_k$ (where
$t_0 = t$ in the case of scalar types).
% \end{itemize}

\mysubparagraph{Types as set of values} %
In addition, labels also constrain the values in memory: e.g., a
cell $a$ labeled with $\mathtt{(bar_0)*}$ should only contain
addresses of cells $b$ whose label is $\mathtt{(bar_0)}$ (in the example, both
cells $\mathtt{0x01}$ and $\mathtt{0x05}$ must contain the address $\mathtt{0x02}$).
This makes use of the second interpretation of types, as sets of
values; this interpretation depends on a labeling
$\mathscr{L}$. Formally, we define the set of values
$\llparenthesis t \rrparenthesis_{\mathscr{L}} \in
\mathcal{P}(\mathbb{V})$ of a type $t \in \mathbb{T}$ as
follows: \vspace{-2mm}
\begin{align*}
  \llparenthesis \mathtt{int} \rrparenthesis_{\mathscr{L}} =   \llparenthesis \mathtt{word} \rrparenthesis_{\mathscr{L}}  &= \mathbb{V} \\
  \llparenthesis t* \rrparenthesis_{\mathscr{L}} &= \{ a \in \mathbb{A} : \mathscr{L}(a) \sqsubseteq t \} \\
  \llparenthesis s_k \rrparenthesis_{\mathscr{L}} &= \bigcap_{t: s_k \sqsubseteq t} \llparenthesis t \rrparenthesis_{\mathscr{L}}
\end{align*}

Here is an example explaining how the values that a memory cell can hold are constrained. The address $\mathtt{0x05}$
is labeled by $\mathtt{bar_3}$. The set values that an address of type
$\mathtt{bar_3}$ can hold is $\llparenthesis \mathtt{bar_3} \rrparenthesis_{\mathscr{L}} = \llparenthesis \mathtt{foo_1} \rrparenthesis_{\mathscr{L}} = \llparenthesis \mathtt{bar_0}* \rrparenthesis_{\mathscr{L}} = \{ \mathtt{0x02} \}$, which matches the content of $m$ at address $\mathtt{0x05}$.

\mysubparagraph{Labeling for adjacent bytes} %
We require labelings $\mathscr{L}$ to be consistent with pointer
arithmetic, i.e. contiguous bytes belonging to the same structure
should have matching labels. For instance, from the facts that address
$\mathtt{0x03}$ has label $\mathtt{bar_1}$ and that the size of
\texttt{bar} is 4 bytes, we know that address $\mathtt{0x03+2=0x05}$ has
label $\mathtt{bar_{1+2}} = \mathtt{bar_3}$. Formally, this is
expressed with the following constraint:
$\forall k \in \mathbb{N} : \forall a \in \mathbb{A}:$
\[ \mathscr{L}(a) = s_o\ \land\ 0 \le o +k < \operatorname{size}(t)\ \ \Longrightarrow\ \ \mathscr{L}(a+k) = s_{o+k} \]

\mysubparagraph{Types and aliasing} %
An important property of our system is that two addresses whose labels
are not in a subtyping relationship do not alias. For instance, two
addresses whose labels are $\mathtt{bar_0}$ and $\mathtt{bar_1}$ may
not alias, but two addresses with labels $\mathtt{bar_2}$ and
$\mathtt{foo_0}$ may alias.
% for any $\mathscr{L}$,
% $\llparenthesis \mathtt{bar_0*} \rrparenthesis_{\mathscr{L}} \cap \llparenthesis
% \mathtt{bar_1*} \rrparenthesis_{\mathscr{L}} = \emptyset$, but
% $\llparenthesis \mathtt{bar_0*} \rrparenthesis_{\mathscr{L}} \cap \llparenthesis
% \mathtt{int*} \rrparenthesis_{\mathscr{L}} \ne \emptyset$.
This property is used by our analysis to ensure that some parts of
memory (like page tables) are not modified by a memory store.

\begin{figure}[t]
  \centering\begin{lstlisting}
    typedef struct { foo *next; bar *data; } foo;
    typedef struct { int x; int y; foo f; } bar;
    foo u; bar v;
  \end{lstlisting}
  \vspace{-1mm}
  (a) Type and globals definitions in C

  \vspace{-1mm}
  \makebox[8cm]{\hrulefill}

  \vspace{1mm}

    \centering\begin{tikzpicture}[xscale=1.7,yscale=0.7]
%      \node (word) at (0.5,1) {$\mathtt{all}$};
    \node (word) at (0.5,1) {$\mathtt{word}$};      
    \node (foostar0) at (1,0) {$\mathtt{foo}_0{*}$};
    \node (barstar0) at (2,0) {$\mathtt{bar}_0{*}$};        

    \node (intstar) at (-0.5,0) {$\mathtt{int}$};
    
    \node (foo0) at (1,-1) {$\mathtt{foo}_0$};
    \node (foo1) at (2,-1) {$\mathtt{foo}_1$};        

    \node (bar0) at (-1,-2) {$\mathtt{bar}_0$};
    \node (bar1) at (0,-2) {$\mathtt{bar}_1$};        
    \node (bar2) at (1,-2) {$\mathtt{bar}_2$};
    \node (bar3) at (2,-2) {$\mathtt{bar}_3$};

    \draw (intstar) edge (word);
    \draw (intstar) edge (bar0);
    \draw (intstar) edge (bar1);
    \draw (word) edge (foostar0);
    \draw (word) edge (barstar0);        
    \draw (foostar0) edge (foo0);
    \draw (barstar0) edge (foo1);

    \draw (foo0) edge (bar2);
    \draw (foo1) edge (bar3);    
  \end{tikzpicture}
  
  \vspace{-1mm}
  (b) Subtyping relations derived from the C definitions
%  \vspace{-8mm}

%  \subcaption{Subtyping relations}

  \vspace{-1mm}
  \makebox[8cm]{\hrulefill}

  \vspace{1mm}

  \centering\begin{tikzpicture}[yscale=0.8]

    \draw[] (0,0) rectangle (6,1);
    \draw[] (1,0) -- (1,1);
    \draw[] (2,0) -- (2,1);
    \draw[] (3,0) -- (3,1);
    \draw[] (4,0) -- (4,1);
    \draw[] (5,0) -- (5,1);    

    \node[below=-2pt,align=center] at (0,0) {\tt \scriptsize 0x00\\[1ex] \tt \&u};
    \node[below=-2pt] at (1,0) {\tt \scriptsize 0x01};
    \node[below=-2pt,align=center] at (2,0) {\tt \scriptsize 0x02\\[1ex] \tt \&v};
    \node[below=-2pt] at (3,0) {\tt \scriptsize 0x03};
    \node[below=-2pt] at (4,0) {\tt \scriptsize 0x04};
    \node[below=-2pt] at (5,0) {\tt \scriptsize 0x05}; 

    \node at (0.5,0.5) {\tt 0x04};
    \node at (1.5,0.5) {\tt 0x02};
    \node at (2.5,0.5) {\tt 0x00};
    \node at (3.5,0.5) {\tt 0xff};
    \node at (4.5,0.5) {\tt 0x00};
    \node at (5.5,0.5) {\tt 0x02};

    \draw[-stealth,thick] (0.5,0.8) node[circle,inner sep=1pt,fill=black] {}  |- (1,1.45) -| (4,1);
    \draw[-stealth,thick] (1.5,0.8) node[circle,inner sep=1pt,fill=black] {}  |- (1.5,1.3) -| (2,1);
    \draw[-stealth,thick] (4.5,0.8) node[circle,inner sep=1pt,fill=black] {}  |- (1.5,1.6) -| (0,1);
    \draw[-stealth,thick] (5.5,0.2) node[circle,inner sep=1pt,fill=black] {}  |- (2.5,-0.6) -| (2,-0.3);

    \node at (7.3,0.5) {\small
      $\left[\!\!\!\begin{array}{l}
         \mathtt{0x00}\mapsto\mathtt{foo}_0\\
         \mathtt{0x01}\mapsto\mathtt{foo}_1\\
         \mathtt{0x02}\mapsto\mathtt{bar}_0\\
         \mathtt{0x03}\mapsto\mathtt{bar}_1\\
         \mathtt{0x04}\mapsto\mathtt{bar}_2\\
         \mathtt{0x05}\mapsto\mathtt{bar}_3         
       \end{array}\!\!\!\right]$
      };

      \node at (4,-1.3) {(c) Example concrete memory ($m$) and labeling ($\mathscr{L}$)};
      
    % \begin{scope}[shift={(5,0}]
    %   \[
    %     \texttt{0x00}: \mathtt{foo}, 0
    %     % \begin{array}{rl}

    %     %   \texttt{0x00} &: \mathtt{foo}, 0\\          
    %     % \end{array}
    %   \]
    % \end{scope}
    \end{tikzpicture}
    
    \vspace{-2mm}    
    \caption{Typing a small memory}

    \label{fig:typing-example}

\end{figure}
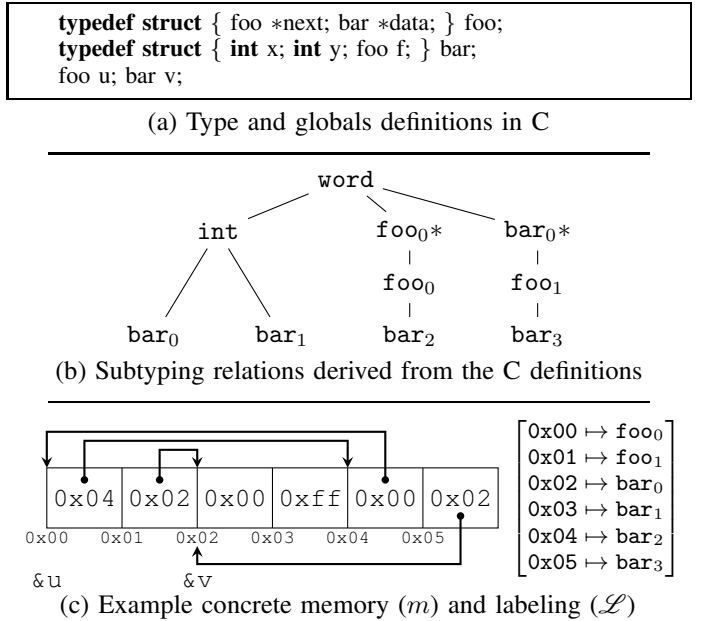

\subsection{The type-based weak shape abstract domain $\mathbb{T}^\sharp$}

We build our abstract domain $\mathbb{T}^\sharp$ upon preservation of
the labeling of memories: modifying a memory $m$ for which
$\mathscr{L}$ is a labeling, results in a new memory $m'$, for which
$\mathscr{L}$ is still a labeling. This is a useful property: it means
that cells holding pointers to \texttt{foo} structures will always
contain those, and cannot point to other data structures (such as page
tables). For instance, if there are no memory stores to addresses
whose label is a supertype of $t$, we have proved that addresses whose
label is $t$ are read-only.

The reason why the type domain $\mathbb{T}^\sharp$ is an efficient
shape abstraction is that it
% does not track the contents of memory at
% all. Instead, it tracks only the types of values \emph{outside} of the
% memory, which may contain addresses to the memory. Applied to our
% kernel, we
does not track the contents of addresses in the user image
(the parameterized part, $\mathbb{A}_P$) at all. It only tracks the types of the values in
registers $\mathcal{R}$ and memory cells of the rest of the kernel
(the fixed part, $\mathbb{A}_F$), ensuring that they point to appropriate addresses in
$\mathbb{A}_P$.

% We only know that $\mathbb{A}_P$ is well-typed,
% and that pointers in registers $\mathcal{R}$ and memory cells of the
% rest of the kernel $\mathbb{A}_F$ points to appropriate addresses in
% $\mathbb{A}_P$.

% When applied to the partitioning of addresses in a kernel
% (Figure~\ref{fig:partitioning-memory-regions}), the type domain
% $\mathbb{T}^\sharp = (\mathbb{A}_{F} \uplus \mathcal{R}) \to
% \mathbb{T}$ tracks both the types of the contents of the registers
% $\mathcal{R}$ and of the addresse in the fixed-address part of memory
% $\mathbb{A}_F$. We don't know anything about the contents of the
% parameters part of the memory $\mathbb{A}_P$; we only know that a
% labeling for the memory $\mathscr{L}$ exists, and pointer in
% $\mathbb{A}_{F} \uplus \mathcal{R}$ to $\mathbb{A}_P$ holds
% appropriate addresses in $\mathbb{A}_P$.

Thus, our type domain
$\mathbb{T}^\sharp = (\mathbb{A}_{F} \uplus \mathcal{R}) \to
\mathbb{T}$ consists in a mapping, tracking the type of the content of
each register or fixed-address memory cell. The meaning this abstract domain is given using its \emph{concretization function}~\cite{cousot1977abstract} $\gamma_{\mathbb{T}^\sharp}$, which is a mapping from abstract values to the set of states they represent:
\begin{multline*}
  % \gamma_{\mathbb{T}^\sharp}:{\mathbb{T}^\sharp} \to \mathcal{P}((\mathbb{A} \to \mathbb{V}) \times (\mathcal{R} \to \mathbb{V})) \\
  \gamma_{\mathbb{T}^\sharp}:{\mathbb{T}^\sharp} \to \mathcal{P}(\mathbb{S}) \\  
  \arraycolsep=2pt  
  \gamma_{\mathbb{T}^\sharp}(t^\sharp) = \{ (\textsl{mem},\textsl{regs}) \mid
  \begin{array}[t]{rl}
    \exists \mathscr{L}: & \forall r \in \mathcal{R}\ \,:  \textsl{regs}[r] \,\;\in \llparenthesis t^\sharp[r] \rrparenthesis_\mathscr{L} \\    
    \land & \forall a \in \mathbb{A}_F: \textsl{mem}[a] \in \llparenthesis t^\sharp[a] \rrparenthesis_{\mathscr{L}} \\
    \land & \mathscr{L} \textrm{\ is a labeling for\ } \textsl{mem} \ \}  
%    \land & \forall a \in \mathbb{A} : \\ & \quad \mathscr{L}[a] \textrm{\ is a \bf struct} \Rightarrow a \in \mathbb{A}_P \}
  \end{array}
\end{multline*}

In practice, the main operations of the domain consist in handling
pointer arithmetic (by communicating with the numeric domain) to
precisely track offsets inside structures; using type information
to retrieve the type of pointers for load operations (e.g., loading
from an address of type $\mathtt{foo_1*}$ returns an address of type
$\mathtt{bar_0*}$); and verifying the preservation of the type upon
stores (e.g., storing a value of type $\mathtt{int}$ into an address
of type $\mathtt{foo_1*}$ makes the analysis fail with an error).

\subsection{Additional features}
\label{sec:additional-features}

Several extensions are required in the actual analysis to successfully
run our case study:
\begin{itemize}
\item We need to handle numerical properties over scalar types; for instance values contained in \texttt{flags} field in
  Figure~\ref{fig:a-tiny-kernel} should always have the ``UNPRIVILEGED''
  bit set, or  \texttt{data\_segment} does not intersect with
  the range of kernel addresses (Figure~\ref{fig:appendix:config-toy} p.\pageref{fig:appendix:config-toy}); 

\item Another important extension is to handle null pointers. 
%: in  general, the C type $s*$ may be a pointer to a structure $s$ or have  the \texttt{null} value. 
Our type system distinguishes
  {\it never-null} and {\it maybe-null} pointer types (the former being
  a subtype of the latter)  and issues alarms on null pointer dereferences; 

\item Lastly, the domain is extended to support arrays 
  in order  to implement bound checking  and to support type annotations 
   expressing  the fact that the length of an array is
  contained somewhere else in memory.
\end{itemize}

\noindent Appendix~\ref{app:sec:annotation-language} (p.\pageref{app:sec:annotation-language}) presents our complete annotation language.  
% language we use to define the types used by the domain, that uses all the above features.

% of the  definitions with all these features \xxx{Complete type definition, with al features, is found in appendix.}

% \subsection{The initial state checker}

% \maybe{Peut-etre que ca va dans la case-study plutot?}

% The correct behavior of the kernel depends on the fact that the
% parameter region is well-formed, e.g., that pointers to \texttt{task}
% structure point to these structures and not to page table entries. The
% typed abstract domain presented in
% Section~\ref{sec:type-abstract-domain} makes sure that a well-formed
% parameter region remains well-formed, but we still have to ensure that
% this region is initially well-formed. There are several ways to
% implement this checking: during kernel initialization, by the
% bootloader, or offline; the latest being the most efficient
% solution. To that end we implemented an simple \emph{initial state
%   checker}, which just consists in checking that a labeling
% $\mathscr{L}$ exists for the initial memory provided by the
% application (see Section~\ref{sec:typed-memories}).

\section{Case study \& Experimental Evaluation}
\label{sec:case-study}

We seek to answer the following Research Questions:
\begin{itemize}[leftmargin=2.5em]
\item[\textbf{RQ1:}]\textbf{Effectiveness} Can \binseccodex{} automatically
  verify absence of privilege escalation on a real microkernel?
\item[\textbf{RQ2:}] \textbf{Internal evaluation} What are the respective impacts of the
  different elements of our analysis method?
\item[\textbf{RQ3:}] \textbf{Scalability and analyzer performance} How  does
  \binseccodex{} scale? % and how does it compare to other OS verification methods?
\item[\textbf{RQ4:}] \textbf{Genericity} Does \binseccodex{} work on different kernels, hardware architectures and toolchains? % and how does it compare to other OS verification methods?
  
\end{itemize}

\subsection{Experimental setup}

\mysubparagraph{Implementation} Our static analysis technique 
%and weak shape abstract domains 
has been implemented in 
\binseccodex{}, a plugin of the   \binsec~\cite{david_binsec/se:_2016} framework for binary-level semantic analysis. We reuse the intermediate-representation lifting and executable file parsing of the platform, but reimplement a whole static analysis~on~top of it, with adequate domains, including weak shapes. Development is done in OCaml,  
the plugin counts around  41~kloc.

\mysubparagraph{Case study} We consider an industrial case study: the microkernel of AnonymOS, a part
of an industrial solution for implementing security- and
safety-critical real-time applications, used in industrial automation,
automotive, aerospace and defense industries. The kernel is being
developed by AnonymFirm\Ldash{}an
SME whose engineers are not formal method
experts\Rdash{}using standard compiler toolchains. The system is
parameterized: the kernel and the user tasks are compiled separately
and both are loaded at runtime by the bootloader, as explained in Section~\ref{sec:toy}. 
%
%This situation
%provides a strong incentive for non-invasive, automated and
%cost-effective verification of security properties of the AnonymOS
%kernel executable.

We have analyzed a port of the kernel to a \emph{4-cores} ARM
Cortex-A9 processor with \emph{ARMv7} instruction set. It relies on the ARM MMU for memory
protection (\emph{pagination}). The executable file contains  1,746
instructions for the boot code  and 796 instructions for the runtime, totalizing 2,542
instructions\Dash{}in line with the general practice of critical embedded microkernels.  We have analyzed two versions: 
\begin{itemize}[leftmargin=4mm]
\item \textbf{\textsc{beta}}, a preliminary version where  
%  that we received while the port  was being made. 
we found a \mbox{vulnerability}; % in this code; 

\item \textbf{\textsc{v1}}, a more polished version where AnonymFirm fixed the vulnerability
  and removed some debug code.  
\end{itemize}
Also, AnonymFirm provided us with a sample user image.

\mysubparagraph{Experimental conditions} %\todo{Other?}
We performed our formal verification
completely independently from AnonymFirm activities. In particular, we
never saw the source code of the kernel, and our interactions with
AnonymFirm engineers were limited to a general presentation of AnonymOS
features and caracteristics of the processor. We ran our analysis on a
laptop with an Intel Xeon E3-1505M 3~GHz CPU with 32 GB RAM.

\subsection{Effectiveness (RQ1)}

\mysubparagraph{Protocol} The goal of our first experiment is to
evaluate the \emph{effectiveness} of the approach, measured by: 
(1) the
fact that the {\it method indeed  succeeds} in computing a non-trivial invariant for the
whole system, i.e., computes an  {invariant under  precondition} for the kernel runtime and checks
that the user tasks establish the precondition;  
(2) the \emph{precision}
of the analysis, measured by the number of \emph{alarms}
(i.e. assumptions that the analyzer cannot prove); 
(3) the \emph{effort} 
necessary to setup the analysis, measured by the number of lines of
manual annotations; and 
(4) the \emph{performance} of the analysis,
measured in CPU time and memory utilization. 

We consider both kernel versions and two configurations (i.e., sets of shape annotations): 
\begin{itemize}
\item  \textbf{\textit{Generic}} contains types and parameter invariants which must hold for all
legitimate user image;  % that do not hook unknown code into the kernel; 

\item \textbf{\textit{Specific}} further assumes that the stacks of all user
  tasks in the image have the same size. 
  %This assumption may not hold when user
  %tasks are optimized for space using an analysis of stack usage, but
  This is the default for \asterios{} applications, and it holds on our case
  study. 
   
\end{itemize}
% \noindent\\
% \noindent 

\begin{table}[tbp]
  \centering
  \caption{Main verification results}
  \label{tab:main-verif-results}
  \newcolumntype{"}{@{\hskip\tabcolsep\vrule width 0.8pt\hskip\tabcolsep}}  
  \begin{tabular}{|p{12mm}|p{8mm}|c|c"c|c|}
    \cline{3-6}
    \multicolumn{2}{c|}{} & \multicolumn{2}{c"}{\emph{Generic} annotations} & \multicolumn{2}{p{23mm}|}{\hspace{-1mm}\makebox[23mm]{\emph{Specific} annotations\phantom{\rule{0pt}{1em}}}} \\
    \hline
     \multirow{2}{*}{\hspace{0mm}\makecell[r]{\# shape\\annotations}} &\makebox[8mm]{generated}& \multicolumn{4}{c|}{1057\phantom{\rule{0pt}{1em}}} \\
    \cline{2-6}
    & \phantom{\rule{0pt}{1em}}{manual} & \multicolumn{2}{c"}{57 (5.11\%)} & \multicolumn{2}{c|}{58 (5.20\%)}\\    
    \hline\hline
    \multicolumn{2}{|c|}{\makecell[c]{Kernel version}\phantom{\rule{0pt}{1em}}} & \makecell{\textsc{beta}} & \makecell{\textsc{v1}} & \makecell{\textsc{beta}} & {\makebox[1mm]{\textsc{v1}}}\\
    \hline
    \multirow{2}{*}{\hspace{-1.3mm}\makecell[r]{invariant\\computation}} & \makecell[r]{status} & \phantom{\rule{0pt}{1em}} \tick & \tick & \tick & \tick \\
    \cline{2-6}
    & \makebox[8mm]{time (s)} & \phantom{\rule{0pt}{1em}} 647 & 417 & 599 & \makebox[1mm]{406} \\
    \hline
    \multicolumn{2}{|c|}{\makecell[c]{\# alarms in runtime}} & \makecell{\scriptsize 1 \textbf{true error}\\\scriptsize 2 false alarms} & \makecell{\scriptsize 1 false\\\scriptsize alarm} & \makecell{\scriptsize 1 \textbf{true error}\\\scriptsize 1 false alarm} & \textbf{\scriptsize 0} \tick  \\
    \hline
    \multirow{2}{*}{\hspace{1.5mm}\makecell[r]{user tasks\\checking}} & {\makecell[r]{status}} & \phantom{\rule{0pt}{1em}} \tick & \tick & \tick & \tick \\ \cline{2-6}

    & \makebox[8mm]{time (s)} & \phantom{\rule{0pt}{1em}} 32 & 29 & 31 & 30 \\[1pt]
    \hline
    \hline

%    \multirow{2}{*}{\hspace{1.5mm}\makecell[r]{APE}} & {\makecell[r]{status}} & \phantom{\rule{0pt}{1em}} \wrongtick & \tick & \wrongtick & \tick \\ \cline{2-6}
   
    \multicolumn{2}{|c|}{Proves APE?} & \phantom{\rule{0pt}{1em}} N/A  & \textcolor{orange!50!green!50!black!50}{\large $\mathbf{\sim}$}  & N/A  & \tick  \\[1pt]
    \hline
  \end{tabular}
\end{table}

%% Manually moved because double paged are always deferred with Latex.
\begin{table*}[tbp]
  \centering
    \caption{Impact of the methodology.}
  \label{tab:internal-evaluation}
%  \newcolumntype{"}{@{\hskip\tabcolsep\vrule width 0.8pt\hskip\tabcolsep}}
  \newcolumntype{"}{@{\hskip\tabcolsep\vrule width 0.8pt}}
  \begin{tabular}{|l"r|r"r|r"r|r"r|r"r|r"r|r|}
    % \Xhline{2\arrayrulewidth}
    \hline
    Annotations & \multicolumn{2}{c"}{\makecell{No annotation}} & \multicolumn{2}{c"}{\makecell{Generated}} & \multicolumn{2}{c"}{\makecell{Minimal}} & \multicolumn{2}{c"}{\makecell{Generic}} & \multicolumn{2}{c"}{\makecell{Specific}} & \multicolumn{2}{c|}{\makecell{Dedicated}}\\
    % & \multicolumn{1}{c}{(init)} & \multicolumn{1}{c"}{(runtime)}     & \multicolumn{1}{c}{(init)} & \multicolumn{1}{c"}{(runtime)}    & \multicolumn{1}{c}{(init)} & \multicolumn{1}{c"}{(runtime)}    & \multicolumn{1}{c}{(init)} & \multicolumn{1}{c"}{(runtime)}    & \multicolumn{1}{c}{(init)} & \multicolumn{1}{c|}{(runtime)} \\
%    \Xhline{2\arrayrulewidth}    
%    \hline
    \hline
    \hline
    Generated annotations (total)     & \multicolumn{2}{r"}{0} & \multicolumn{2}{r"}{1057} & \multicolumn{2}{r"}{1057} & \multicolumn{2}{r"}{1057} & \multicolumn{2}{r"}{1057} & \multicolumn{2}{r|}{1057}\\    
    Manual annotations (total)     & \multicolumn{2}{r"}{0} &
    \multicolumn{2}{r"}{0} & \multicolumn{2}{r"}{10} & \multicolumn{2}{r"}{57} &
    \multicolumn{2}{r"}{58} & \multicolumn{2}{r|}{62}\\
    \hline
    \multicolumn{13}{c}{} \\[-6pt]     \cline{2-13}
%    \\[-3pt]

     \multicolumn{1}{c"}{}                   &\multicolumn{1}{c|}{boot}& \multicolumn{1}{c"}{runtime}     & \multicolumn{1}{c|}{boot} & \multicolumn{1}{c"}{runtime}    & \multicolumn{1}{c|}{boot} & \multicolumn{1}{c"}{runtime}    & \multicolumn{1}{c|}{boot} & \multicolumn{1}{c"}{runtime}    & \multicolumn{1}{c|}{boot} & \multicolumn{1}{c"}{runtime} & \multicolumn{1}{c|}{boot} & \multicolumn{1}{c|}{runtime} \\
                        % \Xhline{2\arrayrulewidth}
    \hline
    Analysis time (s)                           & \wrongtick %5.5 
& \multicolumn{1}{c"}{N/A}
                                                & \wrongtick %9.2 
& \multicolumn{1}{c"}{N/A}
    & 342 & 394 & 195 & 222 & 187 & 219 & 151 & 203\\% \Xhline{2\arrayrulewidth}
    \hline
    Total number of alarms     & \multicolumn{1}{c|}{N/A} &
    \multicolumn{1}{c"}{N/A} & \multicolumn{1}{c|}{N/A} & \multicolumn{1}{c"}{N/A} & 85 & 13 & 60 & 1 & 59 & 0 & 43 & 0\\
    \hline
     User tasks checking     & \multicolumn{2}{c|}{N/A} & \multicolumn{2}{c"}{N/A} & \multicolumn{2}{c|}{\tick} & \multicolumn{2}{c"}{\tick} & \multicolumn{2}{c|}{\tick} & \multicolumn{2}{c"}{\tick} \\ 
    \hline
%    {pointer arithmetic out of bound}          & \multicolumn{1}{c|}{\multirow{5}{*}{N/A}} & \multicolumn{1}{c"}{\multirow{5}{*}{N/A}} & \multicolumn{1}{c|}{\multirow{5}{*}{N/A}} & \multicolumn{1}{c"}{\multirow{5}{*}{N/A}} & 29 &   6  &  16 &   1 &   6 &   0 \\
%    {out of bound array access}                 &                        &          &            &             & 16 &   3  &   9 &   0 &   2 & 0   \\
%    {store does not preserve shape constraints} &                        &          &            &             &  6 &   0  &   4 &   0 &   4 & 0   \\
%    {invalid memory protection modification}    &                        &          &            &             &  0 &   0  &   0 &   0 &   0 & 0   \\
%    {other alarms}                              &                        &          &            &             & 34 &   4  &  31 &   0 &  31 & 0   \\
    % \Xhline{2\arrayrulewidth}
%    \hline
  \end{tabular}

\end{table*}

\mysubparagraph{Results} The main results are given in
Table~\ref{tab:main-verif-results}. The \textbf{\textit{generic}} annotations consist in only 
57 lines of manual annotations, in addition to 1057 lines that were automatically
generated (i.e.\@ 5~\% of manual annotations). When analyzing the \textbf{\textsc{beta} version} with these annotations,
only {\it 3 alarms} are raised in the runtime:

\begin{itemize}
\item One is a {\bf true vulnerability}: in the supervisor call entry
  routine (written in manual assembly), the kernel extracts the system
  call number from the opcode that triggered the call, sanitizes it
  (ignoring numbers larger than 7), and uses it as
  an index in a table to jump to the target system call function; but
  this table has only 6 elements, and is followed by a string in
  memory. This off-by-one error allows a \texttt{svc 6} system call to
  jump to an unplanned (constant) location, which can be
  attacker-controlled in some user images. The error is detected as
  the target of the jump goes to a memory address whose content is not
  known precisely, and that we thus cannot~decode; 

%  \todo[inline]{An other programming mistake there was that
%    the verification was performed using a signed comparison, but our tool found out that this had no security impact
%    because the sign is cleared earlier using a bitmask.
%  }

% % 120000b0:	e24e0004 	sub	r0, lr, #4
% % 120000b4:	e5900000 	ldr	r0, [r0]
% % 120000b8:	e3c004ff 	bic	r0, r0, #-16777216	; 0xff000000
% % 120000bc:	e3500007 	cmp	r0, #7
% % 120000c0:	aa000020 	bge	12000148 <svc_error>

\item One is a false alarm caused by {\it debugging code} temporarily
  violating the shape constraints: the code writes the constant
  \verb|0xdeadbeef| in a memory location that should hold a pointer to a user stack
  (yielding an alarm as we cannot prove that this
  constant is a valid address for this type), and that
  memory location is always overwritten with a correct value further
  in the execution;

\item  The last one is a false alarm caused by an imprecision in our 
  analyzer when user stacks can have different sizes.
\end{itemize}

When analyzing the  \textbf{\textsc{v1} version}, the first two alarms disappear (no new
alarm is added). Analyzing the kernel with the \textbf{\textit{specific}} annotations   %(1 more annotation related to user stack sizes) 
makes the last alarm disappear. In all cases user tasks checking succeeds. 
% (and still succeeds in checking our sample user tasks). 

\smallskip 

\textit{Analyzing the \textbf{\textsc{v1}} kernel with the \textbf{\textit{specific}}
annotations allows to reach 0 alarms, meaning that we have \emph{a
  fully-verified invariant} and a proof of absence of privilege escalation.} 

\smallskip 

%The time needed to compute the invariant, and the time needed to
%interpret the boot code with the parameters and to check that the
%invariant holds, is almost unaffected by the small variations of
%kernel and shape configurations, 

Computation time is  {\it always very  low}: less than 11 minutes for 
the static analysis and 35 seconds for user tasks checking. %invariant checking.  

\mysubparagraph{Conclusions} Experiments show that \emph{verifying
  absence of privilege escalation of an industrial microkernel using
  only fully-automated methods is feasible} (albeit with a very slight
amount of manual annotations). Especially:

\begin{itemize}
\item The analysis is {\it effective}, in that it identifies real errors in the code, and verifies
  their absence once  removed; 

\item The analysis is {\it very precise}, we manage  to reach  0 false alarm 
on the correct code, and we had no more than 2 false alarms on each configuration of the analysis; 

\item The \emph{annotations burden is very small} (58 simple
  lines), as most of the annotations are extracted fully
  automatically; 

\item Finally, the \emph{analysis time}, for a kernel whose size is
  typical for microkernels, \emph{is very small} (between 406 and 647 seconds). 
Analysis time would be even smaller (86 seconds) if
the system had used a single core.     
\end{itemize}

\noindent An {\bf additional finding} is that verifying low-level code such as the
\emph{assembly parts of the kernel is essential}: this code is small
but prone to errors, as witnessed by the one we found.

%  Our analysis also computes other useful properties (absence of runtime error, respect of call/return conventions) 
%  and relevant information that can
%  be verified against the expectations of kernel developers  
%  % the fact that the kernel is free of runtime errors  (i.e., cannot crash), 
%  %no runtime error, respect of call/return conventions, 
%    (control-flow graph,  read and written addresses, memory shared between cores, etc.) %the memory areas shared between the processors, 
%  

% the AnonymFirm developers were also interested by the
% invariants computed by our analysis: the fact that the kernel is free
% of runtime errors (i.e. cannot crash), but also to confirm their
% expectations about the control-flow graph, the set of addresses read
% and written by the kernel, the memory areas that are shared between
% the processors, etc.

\subsection{Evaluation of the methodology (RQ2)}

\mysubparagraph{Protocol} The goal of this experiment is to evaluate
our 3-step methodology (Section~\ref{sec:general-method}),  in
particular 
(1)  whether our shape domain is needed, 
(2) what is the nature and impact of the shape annotations, 
and (3) whether differentiated handling of boot code is mandatory. 

We experiment on the \textbf{\textsc{v1}}
kernel version, and report results for both the boot code and the runtime, using different sets of annotations with an
increasing amount of annotations:
\begin{itemize}
\item \emph{No annotation} (equivalent to having no shape domain); 
\item \emph{Generated} annotations (without any manual annotations);
\item \emph{Minimal} annotations with which the analysis terminate;
\item \emph{Generic} and \emph{Specific} are the annotations defined above;
\item \emph{Dedicated} hardcodes some parameters, such as the number of tasks, with values of the sample user tasks. 
\end{itemize}

\mysubparagraph{Results} Table~\ref{tab:internal-evaluation} shows the
result of this evaluation. The analysis does not succeed in finding an 
 invariant without the shape domain or without manual annotations\Dash{}the analysis is 
too imprecise and aborts in  boot, denoted \wrongtick. 
% it ends up analyzing a write to an unknown
%address, after which it cannot succeed in finding an invariant, denoted \wrongtick. 
%
Only 10 lines of manual annotations are necessary for the analysis to
complete (\textbf{\textit{minimal}}), albeit with many alarms in both boot code and runtime. These annotations mainly %either 
limit the range of array indices to prevent overflows in pointer
arithmetic.
% or enforce that an experimental feature (allowing user
%tasks to hook trusted code in the kernel, which we do not yet support) is not used.
%
The \textbf{\textit{generic}} configuration adds 47 lines 
indicating which pointer types or structure fields may be null, 
which fields  hold array indices, and relating 
array lengths with memory locations holding these lengths. This
configuration eliminates most alarms in the runtime, but 60 alarms
remain in the boot code. The \textbf{\textit{specific}} annotations reach 0 alarms in runtime, but still 59 alarms in boot code. 
%
% The next 38 additional lines of annotation alows to remove almost
% every alarm in the runtime (and many alarms in the initialization),
% without further restricting the parameters supposed to be executable
% by the kernel. All of these annotations were simple, mainly consisting
% in telling which pointer type, or field in a structure, may hold null
% pointers; which is an array index (and must be consistent with the
% array bounds); and to relate the size of arrays with the memory
% locations that hold these sizes. A kernel expert confirmed that this
% configuration is supposed to hold for every legitimate parameters.
%
Even the \textbf{\textit{dedicated}} annotations
%is an attempt to remove the remaining
%alarms  by specializing the analysis to our sample
%user tasks. This did reduce the number of alarms, but we could not
%eliminate them completely. 
cannot   eliminate all alarms in boot code. 

Interestingly, we also found that some of the
invariants in the generated annotations do not hold before boot.

%had we eliminated these alarms, checking the user task on the initial state
%was still impossible, since some of the
%invariants in the shape configuration do not hold before boot. 

% by writing a configuration that would be 

% reduce the number of alarms in the
% initialization, by specializing the analysis to correspond to our
% example application, but we failed at completely eliminating all the
% remaining alarms (especially for the code that constructs memory
% protection table at runtime). Furthermore, we found out that many of
% the shape invariants that are true during the runtime do not hold
% before initialization, i.e. the initialization is necessary to
% establish them.

\mysubparagraph{Conclusions}  \emph{Parameterized
  verification of the kernel cannot be done without the shape
  domain}. The ability to extract the shape configuration from types
is extremely useful, as \emph{95\% of the annotations are automatically
  extracted}, requiring only 57 lines of simple manual
annotations. Finally, \emph{differentiated handling of boot code is
  necessary} as both  the boot code is much harder to analyze than the runtime, and 
the shape invariants holds only after boot code terminates. 

% \begin{itemize}
% \item A shape domain is required to successfully verify a
%   parameterized kernel;
% \item 
% \end{itemize}

% These experiments allows to conclude that \emph{parameterized
%   verification of the kernel cannot be done without a shape
%   domain}. This parameterized verification requires a \emph{small
%   amount of manual annotations}, but these annotations are simple;
% they can be retrieved by hand but it also reasonable to ask kernel
% developers to write them. \emph{Checking the invariant after
%   interpreting the initialization} is mandatory on our kernel, as the
% initialization establishes some of the invariant described in our
% shape constraints. Note that this does not mean that we can skip the
% parameterized abstract interpretation to analyze directly the runtime:
% for instance the AnonymOS kernel creates its interrupt dispatch table
% during initialization, and the runtime cannot be analyzed without
% knowing the value of this table.

\subsection{Scalability and analyzer performance (RQ3)}

\mysubparagraph{Protocol} We now turn to scalability and performance.
% (in term of computing
%resources), and its ability to scale to a large number of user tasks. 
%
(1)~First, we  specialize the \emph{generic} shape
annotations to fix the number of tasks in the system to different
constant values, in order to assess its scalability w.r.t.~the number of tasks.  
(2)~Second, we examine the interpretation
of the boot code and measure the number of instructions executed
deterministically -- i.e. instructions whose inputs and outputs are singleton values. % an instruction is deterministically  executed
% if it is equivalent to performing a sequence of
% writes of a single value to a single address or register 
%(including the program counter)
% This is important regarding checking, as  such instructions can be interpreted very efficiently.
This is important regarding user tasks checking, as such instructions can be analyzed exactly using efficient concrete (non-abstract) interpretation. %\todo{ML:La precision est importante puisqu'on parle de scalabilite, i.e. complexite}
%, in time and space proportional to the real execution on the processor. 
% 
%(3) Finally, we discuss our analysis time compared to reported analysis time from the littarature. 

\mysubparagraph{Results} Results of the first experiment are given
in Table~\ref{tab:scalability-experiment}. As expected, the execution time variations between different numbers
of tasks are very low. Actually,  fixing the number of tasks 
provides only a slight speedup compared to the case where this number is
unknown.

%\vspace{-6mm}
\begin{table}[htbp]
  \centering
  \caption{Scalability experiment}
  \label{tab:scalability-experiment}
  \begin{tabular}{rrrrr}
%    \hline
    Number of tasks & 10 & $10^3$ & $10^7$ & unknown\\
    \hline
    \makecell{Static analysis time (s)} & 380 & 387 & 388 & 406 \\
    \makecell{Static analysis memory (GB)} & 5.3 & 5.6 & 5.6 & 5.8 \\
%    \hline
%    \makecell{Parameters\\ checking} & xs & xs & xs & N/A\\
%   \hline
  \end{tabular}
\end{table}

For the second experiment,  out of
111,491 executed instructions on the boot code, 111,447 (99.96\%) were deterministic. The remaining
instructions were all related to low-level device driver
initialization (reading MMIO values), and were independent from %the value of 
 the user image.

%% Finally, Table~\ref{tab:comparison-os-verification} provides
%% verification times in other OS verification efforts. Comparing these
%% numbers should be done with care, as the technique, kernel, verified
%% properties, and machine used to perform the verification are all
%% different. Nevertheless, we can see that our technique performance is in line with
%% the fastest OS verification methods -- actually, on source code or typed assembly. 

\mysubparagraph{Conclusion} Our static analysis scales well to an arbitrary number of 
user tasks.
%\todo{ML: A remettre, important pour la scalabilite}
%The \emph{parameterized analysis of the runtime is
%  essentially independent from the user tasks running on the system}. 
Moreover, almost every instruction in the boot code is executed
deterministically and can thus be executed using efficient concrete
interpretation. 
% in time proportional to the real execution of the boot code. 
Finally, our technique performance is in line with the fastest
OS verification methods (Table~\ref{tab:comparison-os-verification}).

% From these experiments we can conclude that \emph{parameterized
%   analysis of the system is a very scalable method}, since its
% execution time is virtually independent from the number of tasks in
% the system.

% Most of the boot code can be verified using simple interpretation of
% the boot code, which has a constant overhead compared to the real
% execution of the boot code. The technique is thus also scalable.

% Interpreting the initialization to check the parameters is also
% scalable: as the parts that relate to initializing the parameters is
% deterministic, interpreting this part corresponds to a simple
% interpretation of machine code, a problem that \emph{can be solved in
%   time and memory proportional to a real execution of the
%   initalization}. In other words, parameters checking would be slow
% only if the real initialization would also be slow.

%Inferring invariant is no more costly than merely checking 

% Last, we indicate the verification time of other tools. These tools
% verify invariants that are similar \cite{dam2013machine} or
% encompassed \cite{vasudevan2013design} by absence of privilege
% escalation, on a non-parameterized hypervisors (each running two
% separate tasks), and do not infer invariants, but check invariants
% that are provided. We can see that our analysis tool outperform them
% by a large factor.

\subsection{Genericity (RQ4)}% applicability to OS engineering (RQ4)}

\mysubparagraph{Protocol} We now assess the genericity of our approach. % The goal of this last sequence of
% experiments is to show that our method
% can deal with other kernels or hardware architectures, and handles
% small variations in the machine code.
%
% Final version: tell that it is us.
We apply our tool to EducRTOS,\cite{educrtos} a
small academic OS developed for teaching purpose. It is both a separation kernel (with task
isolation) and a real-time OS. % (featuring real-time scheduling).
Interestingly, EducRTOS differs significantly from AnonymOS: it runs on \emph{x86 (32 bit)}
and memory protection relies on \emph{segmentation}. % instead of pagination.
%, and the user image is linked with the OS at compile-time, rather than at load time. 
%three kernel compile-time options (Round-Robin, Earliest-Deadline-First, or Fixed-Priority scheduling),
We consider 6 variants, 
using two compilers (\texttt{clang-8.0.1} and \texttt{gcc-9.2.0}) and three
different optimization levels
(\texttt{O1}, \texttt{O2}, \texttt{O3}). The code contains between 2.7 and 6.3 kbytes of instructions.

\mysubparagraph{Results} We verify absence of privilege
escalation on all variants, in less than 3 seconds each. The \emph{same} annotations (12 manual lines, 144 automatically extracted) are used to verify all variants. There is \emph{no alarm} in the runtime, and 5 alarms in the boot code, that our method allows to ignore.

% Only xxx lines of annotations (out of xxx) are necessary to perform
% the analysis. In one of the commit, the structure of the user tasks
% checking changed in a way that required to modify the configuration.

% In addition, we had to tell the analyzer the name of the symbol used
% by the kernel to find the description of the user tasks, and give it
% the right to access the VGA memory.

% We found a bug thanks to our tool: a designated initializer was
% missing in a file producing the user image, which lead to a pointer to
% the main protection being non-initalized, and thus this protection
% table was at address 0 in memory. This bug did not prevent the OS to
% work nor to be secure (and could not be found by testing), but was
% writing to an unexpected area of memory. This illustrates the interest
% of our method to verify properties other than APE.

\mysubparagraph{Conclusion} Together with the \asterios{} case study,
these results show that \binseccodex{} can be used in a variety of OS
projects, and is robust to small variations in the OS
binary. % EducRTOS is used to teach operating system classes, and
% we plan on using the tool to automatically validate student projects.

% . Furthermore, the rapidity of the
% verification, and the fact that the user configuration is small and
% does not need to be often modified, make the tool suitable for use in
% a build automation server.

\section{Discussion}
\label{sec:disc-limit}

% \myparagraph{Further reducing the trusted base} Verifying a system
% using its binary executable allows to move its source and build chain
% out of the trusted base, to the extend that only the analyzer and the
% system model remain (note that privilege escalation does not require
% to write a specification, which would also be in the trusted base). It
% would be possible to further reduce the trusted base by using a
% verified static analyzer, albeit by paying an important performance
% penalty~\cite{jourdan2015formally}. The hardware model will always
% have to be trusted; ours is shared with a full-system simulator used
% to certify safety-critical software, and whose compatibility with the
% hardware was thoroughly examined.

% Perhaps the main limitation of our work is the need to trust our
% static analyzer and parameters checker. Even if our analysis builds upon
% well-established theoretical background, it is a large tool that can
% have bugs.  The parameters checker is simple, and it should not be too
% hard to develop a formally-verified version (e.g., in Coq
% \cite{...}). Developping a formally verified analyzer for real-world
% languages is also possible \cite{jourdan2015formally}, but this requires a very
% important effort, and currently the performances of these analyzers
% are slow. Another possibility would be to change our analyzer to emit
% proof certificates that could be checked independently by a prover
% like Coq. This would allow to reduce the trusted base of the analysis
% to that of the small Coq kernel.
 
\mysubparagraph{Threats to validity} 
Perhaps the main limitation of our work is the need to trust our
 static analyzer and user tasks checker. 
%Even if our analysis builds upon
% well-established theoretical background, it is a large tool that can
% have bugs.  
We mitigate this problem the following way. 
Our prototype is implemented as  part of
\binsec{}~\cite{david_binsec/se:_2016}  whose 
robustness have been demonstrated in prior large scale studies%
%on both adversarial code and standardly-compiled code
~\cite{bardin_backward-bounded_2017,recoules_ase_2019,david_specification_2016,FeistMBDP16}. Especially, the
IR lifting part has been exhaustively tested~\cite{lhuillier2019cesar} and positively evaluated in an independent
study~\cite{kaist-bar-workshop}. 
%
%\todo{ML: Citer papier Yves CESAR. Aussi: crosscheck interpretation et full analysis pas clair de l'interet, on dirait plutot un mode commun de defaillance.}
Moreover, the user tasks checker is simple and share many components with the analyzer, allowing to crosschecks results 
between interpretation and full analysis. Finally, results have been crosschecked for consistency between versions, 
and all alarms on runtime have been manually investigated. 
 It  would be possible to further reduce the trust base by using a
 verified static analyzer~\cite{jourdan2015formally}, albeit with an important performance penalty. 
 %The hardware model will always
 %have to be trusted; ours is shared with a full-system simulator used
 %to certify safety-critical software \cite{xxunisim}, a whose accuracy with actual 
 %hardware was thoroughly examined.  

\mysubparagraph{Representativity of the case studies} We have verified 
an industrial kernel   designed for  safety-critical and
hard real-time applications, whose size and complexity are in line with 
the general practice of these fields. While  this kernel is indeed  more restricted than a 
general-purpose OS \cite{walker1980specification,klein2009sel4,gu2016certikos} --  no 
dynamic  task creation nor dynamic re-purpose of memory,   
%This is an important limitation; but note that 
fixing  memory partitioning is
standard in embedded~\cite{ramamritham1994scheduling,muehlberg2011verifying,richards2010modeling,duverger2019gustave,ARINC653,dam2013machine} and 
highly-secure~\cite{rushby1981design} systems.
% \cite{ramamritham1994scheduling,muehlberg2011verifying} or
% safety-critical \cite{ARINC653,richards2010modeling} operating
% systems.

Regarding the verification itself, the analysis %is substantial in that it 
does not sidesteps any major difficulty: 
% that may be encountered when verifying a microkernel: 
the kernel was unmodified,  
it runs concurrently,  
 it initializes itself and creates its
memory protection tables dynamically,  
%its size is in line with systems that aim at a low trusted code base, 
 we have verified all of the kernel code -- including  boot  and the assembly parts,  
%  The fact that we
%perform a parameterized analysis of the kernel, on machine code, adds
%additional complexity.
and  our verification is parameterized.  Finally, we have shown that the analysis works on differing setups (architecture, protection).

% M-x toggle-truncate-lines
% (align-regexp (region-beginning) (region-end) "\\(\\s-*\\)&" 1 1 t)
\begin{table*}[htbp]
  \centering
  \scriptsize
  %\hspace{-1mm}
  \begin{threeparttable}
    \caption{Comparison of kernel verification efforts}
    \label{tab:comparison-os-verification}
    \begin{tabular}{|r|cc|ccccc|cccc|}
      \hline
    \multicolumn{1}{|c|}{\textsc{Verified kernel}}                                                                                 & \multicolumn{2}{c|}{\textsc{Verified property}}      & \multicolumn{5}{c|}{\textsc{Verification technique}} & \multicolumn{4}{c|}{\textsc{Case study}} \\[1.5pt] 
                                                                                     & \makecell{Verified\\property}                        & \makebox[5mm]{\makecell{Implies\\APE?}} & \makecell{Degree of\\automation} & \makecell{Verif.\\Level}                & \makebox[7mm]{\makecell{Parame-\\terized}}    & \makebox[5mm]{\makecell{Multi-\\core}}    & \makebox[7mm]{\makecell{Infers\\invariants}} & \makecell{Manual\\Annotations (LoC)}        & \makecell{Unproved\\code (LoC)}          & \makebox[4mm]{\makecell{Non-\\invasive}} & \makecell{Analysis\\time (s)} \\ \hline\hline 
    This work                                                                        & \makecell{Absence of\\priv. escalation}              & \tick                                   & \makecell{\bf Fully\\\bf automated}      & {\bf Machine}                                 & \tick                          & \tick                   & \tick                         & 58\tick                                          & 0 \tick                                       & \tick                     & \bf 406                           \\ \hline\hline
    XMHF \cite{vasudevan2013design}                                                  & \makecell{Memory\\integrity}                         & \makecell{\wrongtick\tnote{b}}          & \makecell{Semi\\ automated}      & Source                                  & \makecell{\wrongtick\tnote{h}} & \tick                   & \wrongtick                    & \makecell{N/A}                              & \makecell{422  (C)\\+ 388  (asm)}        & \wrongtick                & \bf 76                            \\ \hline
    \"uXMHF \cite{vasudevan2016uberspark}                                            & \makecell{Security\\ properties}                     & \tick                                   & \makecell{Semi\\automated}       & \makecell{Source +\\ assembly}          & \makecell{\wrongtick\tnote{h}} & \wrongtick              & \wrongtick                    & 5,544                                       & 0 \tick                                        & \wrongtick                & 3,739                         \\ \hline
    \makecell[r]{Verve\\Nucleus \cite{yang2010safe}}                                 & \makecell{Type Safety}                               & \tick                                   & \makecell{Semi\\automated}       & \makecell{Assembly}                     & \tick                          & \wrongtick              & \wrongtick                    & 4,309                                       & 0 \tick                                        & \wrongtick                & \bf 272                           \\ \hline
    Prosper \cite{dam2013machine,dam2013formal}                                                    & \makecell{Compliance with\\specification}            & \tick                                   & \makecell{Semi\\automated}       & {\bf Machine}                                 & \makecell{\wrongtick\tnote{i}} & \wrongtick              & \wrongtick                    & 6,400 \tnote{c}                             & 0 \tick                                        & \wrongtick                & \makecell{$\le$ 28,800}       \\ \hline
    \makecell[r]{Baby hyper-\\visor \cite{alkassar2010automated,paul2012completing}} & \makecell{Compliance with\\specification}            & \tick                                   & \makecell{Semi\\automated}       & \makecell{Source+\\assembly}            & \tick                          & \wrongtick              & \wrongtick                    & 8,200 tokens                                & 0 \tick                                       & \wrongtick                & 4,571                         \\ \hline
    Komodo \cite{ferraiuolo2017komodo}                                               & \makecell{Compliance with\\specification}            & \tick                                   & \makecell{Semi\\automated}       & Assembly                                & \tick                          & \wrongtick              & \wrongtick                    & 18,655                                      & 0 \tick                                       & \wrongtick                & 14,400                        \\ \hline 
    \makecell[r]{UCLA Secure\\Unix~\cite{walker1980specification}}                        & \makecell{Compliance with\\specification}            & \tick                                   & \makecell{Manual}                & Source                                  & \tick                          & \wrongtick              & \wrongtick                    &  N/A                                        & 80\%                                     & \wrongtick                & N/A                      \\ \hline
    Kit \cite{bevier1989kit}                                                         & \makecell{Compliance with\\specification}            & \tick                                   & \makecell{Manual}                & {\bf Machine}                                 & \wrongtick                     & \wrongtick              & \wrongtick                    & \makecell{1,020 definitions\\+ 3561 lemmas} & 0 \tick                                        & \wrongtick                & N/A                   \\ \hline    
    $\mu$C/OS-II \cite{xu2016practical}                                              & \makecell{Compliance with\\specification}            & \tick\tnote{a}                          & \makecell{Manual}                & Source                                  & \tick                          & \wrongtick\tnote{f}     & \wrongtick                    & 34,887 \tnote{d}                            & 37\%                                     & \tick                     & 57,600\tnote{e}               \\ \hline 
    SeL4 \cite{klein2009sel4,klein2014comprehensive}                                                        & \makecell{Compliance with\\specification}            & \tick\tnote{a}                          & \makecell{Manual}                & Source\tnote{g}                         & \tick                          & \wrongtick              & \wrongtick                    & 200,000                                     & \makecell{1,200  (C, boot)\\+ 500 (asm)} & \wrongtick                & N/A                   \\ \hline
    CertiKOS \cite{gu2016certikos}                                                   & \makecell{Compliance with\\specification}            & \tick                                   & \makecell{Manual}                & \makecell{Source\tnote{g} +\\ assembly} & \tick                          & \tick                   & \wrongtick                    & 100,000                                     & 0 \tick                                        & \wrongtick                & N/A                   \\ \hline
%    \makecell{ED separation\\kernel \cite{heitmeyer2008applying}} & \makecell{Data separation} & Model \\
%    Barthe 2011
%    \xxx{PSOS}    \\

  \end{tabular}

  % N/A: not available here.

  % Pour bevier: on voit dans le companion tech report qu'il faut recompilé (la c'est pour 16 tâches)
  
  \begin{tablenotes}[para]
  \item[a] Assuming that the proof is completed to cover all the code.
  \item[b] Control flow integrity is assumed.
  \item[c] Generated from a 21,000 lines of HOL4 manual proof.
  \item[d] Plus 181,054 LoC of specification and support libraries.
  \item[e] The reported compilation time includes the support libraries.
  \item[f] The verification is concurrent because of in-kernel preemptions.
  \item[g] The translation to assembly is also verified.
  \item[h] The hypervisor supports a single guest
  \item[i] The hypervisor supports two guests
  \end{tablenotes}

  \end{threeparttable}

\end{table*}

\mysubparagraph{Scope of verification} %
The property that we target, absence of privilege  escalation, 
is weaker 
than, e.g., task separation or full functional correctness
\cite{klein2009sel4,gu2016certikos,bevier1989kit,xu2016practical,ferraiuolo2017komodo,alkassar2010automated,dam2013machine,walker1980specification}. Indeed,
most existing OS verification efforts
(Table~\ref{tab:comparison-os-verification}), if completed, would
imply APE as a byproduct. 
On the other hand, APE is a universal property over OS kernels, is essential to security (Theorem~\ref{th:invariant-implies-noescalation}) and must be 
 proved in any complete formal verification effort. In some systems it is even the main
property of interest~\cite{vasudevan2013design}.

%% Actually, one of our contributions
%% (Theorem~\ref{th:invariant-implies-noescalation}) is defining APE and
%% showing that it is, in practice, the weakest kernel property that can
%% be proved, and that it must be proved in any complete formal
%% verification effort. This universality, and the fact that it does not
%% require a specification, makes it a good target for automated
%% verification.

%% Also note that APE implies security properties that have been targeted in
%% previous work, such as inability to enter supervisor mode
%% \cite{bevier1989kit}, memory integrity 
%% \cite{bevier1989kit,vasudevan2013design} or control flow integrity 
%% \cite{vasudevan2016uberspark}. In some systems it is even the main
%% property of interest~\cite{vasudevan2013design}.

%  in most existing work (with
% \cite{bevier1989kit,yang2010safe,dam2013machine,gu2016certikos} being
% notable exceptions) the verification

\mysubparagraph{Degree of automation} While we do use only
fully-automated methods, we still had to write a small amount of
manual annotations (58 lines for AnonymOS and 12 lines for EducRTOS) to complete our main analysis. These
annotations are not state invariants, but configure the analysis so that state invariants can be inferred. Some annotation is unavoidable in a parameterized verification as the
AnonymOS kernel does not ensure APE for arbitrary user tasks
images. Still, the annotation effort is extremely low compared to prior work. % so proving APE must make a precondition explicit.

%\todo[inline]{Annotations are different in nature: template invariants given as relations between types, but invariants about values are fully inferred.}

%% Finally, it can be argued that it is harder to develop a tool that can
%% prove a property automatically, than to prove it directly. In this work,
%% we restricted ourselves to using only
%% fully-automated methods, but obviously
%% One of our goals was to show that complex properties  can be proved using only fully-automated methods.
%% However, in practice these methods could be used in 
%% combination with manual methods. Indeed,  when verifying an OS kernel most of
%% the work is spent stating and verifying invariants (e.g., 80\% in
%% SeL4 \cite{klein2009sel4}); automated verification can reduce this
%% effort. 

\mysubparagraph{Limits} As already stated, our binary-level static analysis cannot handle dynamic task spawning, dynamic memory re-partitioning, 
code self-modification, code generation and recursion. The last limitation could be overcome with state of the art techniques (possibly at the price of precision), 
the other ones lay at the forefront of program analysis research.

% Thanks to those, the 58 lines of annotations is several order of
% magnitude less than the lines written for existing OS verification
% efforts (Table~\ref{tab:comparison-os-verification}).

% XXX: Aussi, automated meth
% XXX: Invariant complet bien plus gros

% ; these annotations are used to determine the
% ``template'' of the invariant that the shape domain presented in
% Section~\ref{sec:shape-abstract-domain} computes. Note that 58 lines
% of annotation is several order of magnitude less than the lines
% written for existing OS verification efforts. These annotations are
% made necessary because of parametrization: the AnonymOS kernel
% correctness depends on assumptions on the parameters, so its proof
% must make these assumptions explicit. We could extract most of these
% assumptions from the C types, but they are not precise enough to
% precisely describe the required properties.

\section{Related work}

\mysubparagraph{Formal verification of OS kernels} %
Table~\ref{tab:comparison-os-verification} presents a comprehensive overview 
of prior OS verification efforts. 
% according to the verification technique they
%used and the kernel or hypervisor that they have verified. 
We already discussed the scope of verification  in
Section~\ref{sec:disc-limit}. 
Note that most existing works leave unchecked hypotheses about
 the code, assuming for instance control flow
 integrity~\cite{vasudevan2013design}, semantics preservation of the
 compilation~\cite{alkassar2010automated,vasudevan2016uberspark} or 
 correctness of some unverified parts of the code~\cite{xu2016practical,walker1980specification,klein2014comprehensive}\Dash{}in
 particular assembly and boot code. % parts~\cite{klein2014comprehensive}).

%% Also, note that most existing works leaves unchecked hypotheses about
%% the code, assuming for instance control flow
%% integrity~\cite{vasudevan2013design}, semantics preservation of the
%% compilation~\cite{alkassar2010automated,vasudevan2016uberspark},
%% correctness of unverified parts of the code~\cite{xu2016practical,walker1980specification} (in
%% particular the assembly parts~\cite{klein2014comprehensive}). In our
%% work all these hypotheses are fully checked,  we are thus complementary to these
%% verification efforts.

%\subsubsection*{Scope of verification}

\subsubsection*{Fully-automated kernel verification} 
%The differences between manual, semi-automated and
%fully-automated verification techniques have been presented in
%Section~\ref{sec:discovering-verifying-inv}.

We qualified as fully-automated the techniques that are able to infer
invariants automatically. Vasudevan et
al.~\cite{vasudevan2016uberspark} used abstract interpretation in
combination with deductive verification and runtime monitoring to
verify manual annotations.
Despite not inferring the invariants, some verification techniques
feature an advanced level of automation. Vasudevan et
al.~\cite{vasudevan2013design} uses bounded model
checking % , which does not allow inferring invariants)
to verify hand-written proof obligations, with the help of manual
verification statements. Dam et
al.~\cite{dam2013machine,dam2013formal} designed their kernel so that
it requires assertions only at the beginning and end of the kernel,
and most of these assertions are generated automatically from a formal
model.

% XXX: WE count as fully-automated techniques that are able to infer
% invariants automatically. Despite not using a fully-automated
% technique in that sense, some research efforts feature advanced level
% of automation. Vasudevan needs to give only some verification
% statement, and Prosper generates most of their system invariant from a
% formal model.

% Prior efforts based 
% on  fully automated verification include Vasudevan et
% al.~\cite{vasudevan2013design}, where bounded model
% checking  % , which does not allow inferring invariants) 
% is used to verify
% hand-written proof obligations,  and Vasudevan et
% al.~\cite{vasudevan2016uberspark}, where abstract interpretation is
% combined with deductive verification and runtime monitoring to verify
% manual annotations.

\subsubsection*{Machine-level kernel verification} Only two previous efforts 
 target 
  machine code \cite{bevier1989kit,dam2013formal}.  
% where the assembler and linker, which
%determines the addresses of kernel objects and can introduce errors
%\cite{klein2014comprehensive}, are removed from the trusted base. 
Both approaches are based on a  flat memory model, and thus cannot handle 
parameterized number of tasks.  
% and thus the kernel they verify handles only a fixed number of tasks.

\smallskip 

{\it We propose \binseccodex{}, the first fully automated approach able to verify absence of privilege escalation  
on microkernels. Our approach is also the first fully automated OS verification method
 enabling binary-level reasoning and parameterized reasoning. 
%  requires only a very low annotation effort  and is very fast. 
While we focus on the key property of APE, we are confident that the method can be extended to stronger properties such as 
task separation. 
Finally, 
our method can complement existing\Ldash{}more manual\Rdash{}techniques to OS verification, either by automatically inferring parts of the annotations and  discharging 
parts of the proof obligations, or by helping to prove unchecked low-level assumptions. %, especially assembly parts~\cite{klein2014comprehensive}.  
}

\mysubparagraph{Static analysis of machine code} %
Relevant sound techniques were discussed in
Section~\ref{sec:static-analysis-machine-code}.  
Most of them either make no assumption at all (raw binary) \cite{kinder2009abstract,DBLP:conf/vmcai/BardinHV11} 
at the price of precision, or rely on implicit extra assumptions (standardly-compiled code) \cite{reps2010there}.

\smallskip

{\it We pick best practices from existing sound analyses\Dash{}any progress there can be directly reused in our method. Our two main novelties   
 are (1) to make checked explicit assumptions  about data layout described by C types, and 
     (2) to specialize our approach to the typical kernel architecture, computing~an invariant (of the runtime) under precondition (established by boot code)   
 and adding an extra precondition checking step. 
}

\smallskip 

Several existing works drop  soundness and    only  compute  {\it ``best effort invariants''} \cite{KinderK12,FeistMP14}  %\todo{ML:OK pour feist, mais a checker pour kinder. SB : certain}
in order to gain robustness \cite{KinderK12} or to guide latter analysis \cite{FeistMBDP16}. This is not an option for us as we 
look for formal verification of APE. Finally, symbolic execution is widely used on machine code \cite{DBLP:conf/icse/BounimovaGM13,DBLP:conf/icst/BardinH08,DBLP:journals/stvr/BardinH11,DBLP:journals/ieeesp/AvgerinosBDGNRW18,Shoshitaishvili16,david_binsec/se:_2016} 
but aims to find bugs rather than to prove their absence. 
%\todo{ML: Preciser que ca ne gère pas les boucles?}
%\todo{ML: Besoin d'autant de références sur execution symbolique? SB : pas de soucis de nb de refs.  Et de 2 références à Gueb? SB : elles disent pas la meme chose}

%% The main originality
%% of our sound static analysis is to make a checked hypothesis about the code,
%% using the data layout described by C types. We did not need to make
%% other hypotheses \cite{reps2010there}, in particular because the small
%% size of microkernels allows fully context-sensitive analysis. 

%%  An up-to-date binary-level static analysis} We build a 
%% state-of-the-art {\it sound} static analyzer for machine code (Section~\ref{sec:static-analysis}), 
%% picking among the best practices from the literature\cite{reps,vedrine,adel,brauer,kinder,simon}
%%   Interestingly, While prior works is partitioned into {\it raw binary} analysis  
%% (no assumption, adequate to adversarial analysis such as malware but extremely difficult to get precise) 
%% and {\it standardly-compiled code}\footnote{Code produced by a standard development and compiling chain.} analysis 
%% (with hard-coded {\it unchecked} extra assumptions, typically on control flow -- trust an extranal disassembler,  
%% or memory partitioning -- the stack is separated from the rest of memory), our own method 
%% does target standardtly-compiled code but the extra assumptions are explicit  (shape annotations) and  fully checked 
%% by the analysis. 

\mysubparagraph{Fully-automated memory analyses}  
%
%\subsubsection*{Fully-automated memory analyses} 
%Memory-based analyses 
%traditionally belong to one of two extremes. 
Points-to
and alias analyses are fast and easy to setup but are too imprecise
for formal verification, and generally assume that the code behaves
nicely, e.g., type-based alias analyses~\cite{diwan1998tbaa} assume
that programs comply with the strict aliasing rule -- while kernel codes 
often do not conform to C standard \cite{klein2014comprehensive}. 
On the other hand,  shape
analyses~\cite{sagiv1999parametric,chang2008relational} can fully
prove memory invariants, but require heavy parametrization and
are generally too slow to scale to a full microkernel. 

\smallskip 

{\it Our weak type-based
shape abstract domain hits a middle ground: it is fast, precise,  
handles low-level behaviors (outside of the C standard) and requires little
configuration. This is also the first shape analysis performed on machine code.}  

\smallskip 

Marron~\cite{marron2012structural} also describes a
weak shape domain performing only weak updates, but on a type-safe
language with no implicit type casts, pointer arithmetic, nor nested
data structures. 

%prior work
%   that would have performed shape analysis at the machine
%  code level.

\mysubparagraph{Type-based verification of memory invariants} 
Walker et al.~\cite{walker1980specification} already observes in the 1980's 
that reasoning on type invariants is well suited to OS kernel
verification. Several systems build around this idea \cite{yang2010safe,ferraiuolo2017komodo}, leveraging a dedicated typed language.
Cohen et al.~\cite{cohen2009precise} describe a typed semantics for C
with additional checks for memory typing preservation, similar to our own checks on memory accesses.  
% which resembles our checks that memory accesses preserve the typing of memory. 
While they use it in a deductive verification tool for C 
 (to verify an hypervisor \cite{alkassar2010automated}),  
 we build
 an abstract interpreter for machine code.

%\vspace{-5mm}
\section{Conclusion}

Operating system kernels are the keystones of
computer system security, leading to several efforts towards  their formal verification. 
Yet, while these prior works were overall successful, they often require a huge amount of manual 
efforts, and  verification  was often led only at source level, on crafted kernels.      
We focus in this paper on the key requirement that
a kernel should protect itself, coined as absence of privilege escalation,  
and provide a methodology to verify 
it  from the kernel executable only, using fully automated
methods with a very low amount of manual configuration. Our methodology
handles parameterized systems  %where the user tasks are not known,
thanks to a novel type-based weak shape abstract domain. 
The technique has been successfully demonstrated on %the unmodified executable file of 
two embedded microkernels, including an industrial one:  with less than 60 lines of manual annotations, we were able to find a vulnerability in a preliminary version and
 to  verify  absence of privilege escalation in a secure version, 
without any  remaining false alarm.   
\bibliographystyle{ieeetr}
\bibliography{biblio}

\appendix

This appendix contains:
\begin{description}
  \item[A.] The annotation language used to define the types manipulated by the weak
    shape domain;

%  \item Table~\ref{tab:internal-evaluation-with-details} is the detailed version
%    of Table~\ref{tab:internal-evaluation} with a breakdown of the annotation
%    types, as well as the types of the alarms emitted by the analysis.

  \item[B.] The concrete semantics that we use as a model of a hardware
    architecture; 

  \item[C.] the abstract domains used in our analysis.
\end{description}

\subsection{Annotation language}
\label{app:sec:annotation-language}

Annotations take the form of type declarations 
similar to those of the C language, except that numerical properties may be specified on
values and array sizes. The annotation language is presented in Figure~\ref{app:fig:annotation-language}, 
and the shape annotations required for the example kernel of Figure~\ref{fig:a-tiny-kernel} is presented in 
Figure~\ref{fig:appendix:config-toy}. An example of code for an user image is given in Figure~\ref{fig:exemple-initial-configuration}.

{%

%\floatstyle{boxed}
%\restylefloat{figure}

  \newcommand{\gr}[1]{\left\langle#1\right\rangle}

  \begin{figure}[htbp] \footnotesize
    \begin{framed}
      \vspace{-3mm}
      \begin{equation*}
    \begin{array}{r r l}
      \gr{constant} &::=& 1 \mid 2 \mid \cdots \\
      \gr{bitsize} &::=& \gr{constant} \\
      \gr{comp} &::=& \mathtt{==}\ \mid\ \mathtt{!=}\ \mid\ <_u\ \mid\ \leq_u\ \mid\ >_u \\
                    &\mid& \geq_u\ \mid\ <_s\ \mid\ \leq_s\ \mid\ >_s\ \mid\ \geq_s \\
      \gr{binop} &::=& \mathtt{+}\ \mid\ \mathtt{*}\ \mid\ \mathtt{-}\ \mid\ \mathtt{/}_u\ \mid\ \mathtt{/}_s \\
                &\mid& \ll\ \mid\ \gg_u\ \mid\ \gg_s \mid\ \mathtt{\&}\ \mid\ \mathtt{"|"} \\
      \gr{variable} &::=& \mathtt{nb\_tasks} \mid \mathtt{kernel\_last\_address} \mid \cdots \\
      \gr{expr} &::=& \gr{constant} \mid \gr{variable} \mid \mathtt{self}\\
                     &\mid& \gr{expr} \gr{binop} \gr{expr}\\
      \gr{predicate} &::=& \gr{expr} \gr{comp} \gr{expr} \\
                    &\mid& \gr{predicate} \mathtt{or} \gr{predicate} \\
%      \gr{paramdecl} &::=& \gr{param}\ \mathtt{with}\ \gr{predicate} \\
      \gr{type} &::=& \mathtt{int}\!\gr{bit size} \\
                &\mid& \gr{type}\ \mathtt{with}\ \gr{predicate} \\
                &\mid& \gr{label}* \mid \gr{label}? \\
                &\mid& \mathtt{struct\ \{} \gr{type}\star \mathtt{\}} \\
                &\mid& \gr{type}[\gr{constant}\mid\gr{variable}\mid\mathtt{unknown}] \\
%                &\mid& \mathtt{page\_table\_entry} \\
%                &\mid& \mathtt{page\_directory\_entry} \\
      \gr{typedecl} &::=& \mathtt{type}\ \gr{label}\ =\ \gr{type} \\
      \gr{annots} &::=& \gr{typedecl}\star \\
    \end{array}
  \end{equation*}
  \vspace{-3mm}
\end{framed}%
\vspace{-4mm}
%  \end{equation*}\hspace{2mm}}
\caption{Annotation language for configuring the shape domain} \label{app:fig:annotation-language}
\end{figure}
}

{

%\floatstyle{boxed}
%\restylefloat{figure}

\begin{figure}[htbp]

\footnotesize    
% \begin{verbatim}

\begin{lstlisting}[morekeywords={define,type,with,or}]
  #define READ 1
  #define WRITE 2
  #define EXEC 4
  type mpu = int64 with 
        (self | WRITE == 0) or 
        ((self $\gg_u$ 32) $>_u$ kernel_last_addr) or
        (((self $\ll$ 32) $\gg_u$ 32) $<_u$ kernel_first_addr)

  #define UNPRIVILEGED 1
  type flags = int32 with (self | UNPRIVILEGED) != 0

  type Task = struct { int32 int32 flags mpu mpu Task* }
\end{lstlisting}
\vspace{-3mm}
%\end{verbatim}
\caption{Shape annotations configuring the analysis for the example kernel}\label{fig:appendix:config-toy}
\end{figure}

}

\begin{figure}[htbp]

  \lstset{language=C,escapeinside={@}{@},label= ,caption= ,captionpos=b,numbers=none,morekeywords={int32,int64}}
  % numberstyle=\scriptsize,numbers=right, stepnumber=1}
  \hbox{\begin{lstlisting}
void code(void){
  while(true){ print("Hello\n"); asm("software_interrupt"); }
}
void after_code(void) {}

#define STACK_SIZE 256
char stack0[STACK_SIZE], stack1[STACK_SIZE];

Task task0 = {
  .sp = &stack0[STACK_SIZE-sizeof(int)];
  .pc = &code;
  .flags = 0 | UNPRIVILEGED;
  .code_segment =
      &code @$\ll$@ 32 | &after_code | READ | EXEC;
  .data_segment =
      &stack0 @$\ll$@ 32 | &stack0[STACK_SIZE] | READ | WRITE;
  .next = &task1;
};

Task task1 = {
  .sp = &stack1[STACK_SIZE-sizeof(int)];
  .pc = &code;
  .flags = 0 | UNPRIVILEGED;
  .code_segment =
      &code @$\ll$@ 32 | &after_code | READ | EXEC;
  .data_segment =
      &stack1 @$\ll$@ 32 | &stack1[STACK_SIZE] | READ | WRITE;
  .next = &task0;
};
  \end{lstlisting}}%
%\end{lstlisting}
  \caption{Example of code producing a user image for the example kernel}
  \label{fig:exemple-initial-configuration}
\end{figure}

%\todo{Rendre consistent avec les autres figures (notamment appendix)}

%\section*{B -- Semantics of an architecture\\ with memory protection}
\subsection{Semantics of an architecture with memory protection}
\label{app:sec:form-hardw-prot}

Here, we instantiate the formalization of
Section~\ref{sec:formalization-formal-description} by providing a
formal model of a computer using operating system software controlling
the use of hardware privilege. Such a definition is necessary to
define and prove correct a static analysis which computes sets of possible
concrete execution described by this model.

\mysubparagraph{States} %
Values $\mathbb{V}$ are machine integers (e.g. $\mathbb{V} = [0,2^{32}-1]$).
$\mathbb{A} \subseteq \mathbb{V}$ is the set of memory \emph{addresses}.  A
memory $m \in \mathbb{M}$ is just a mapping from addresses to values, i.e.,
$\mathbb{M} = \mathbb{A} \to \mathbb{V}$. We note $\mathcal{R}$ the set of
register names in the system; in the example of
Section~\ref{sec:overview-motivating-example}, this includes for instance
``\texttt{mpu}$_1$'' and ``\texttt{sp}''; or ``\texttt{r0}'' and
``\texttt{r11}'' on ARMv7 processors.

States $\mathbb{S}$ are tuples of
$\mathbb{M} \times \left(\mathcal{R} \to \mathbb{V}\right)$
where each state $s$ has a
memory $s.\textrm{mem} \in \mathbb{M}$, and a mapping
$s.\textrm{regs} \in \mathcal{R} \to \mathbb{V}$ of registers' names to their
values.

Registers $\mathcal{R} = \mathcal{R}_S \uplus \mathcal{R}_U $ are
partitioned into \emph{system registers} $\mathcal{R}_S$ and
\emph{user registers} $\mathcal{R}_U$. In a secure system, only the
kernel is able to modify system registers, using the following
mechanism.

\mysubparagraph{Execution level and privilege} %
States are partitioned between \emph{privileged} (i.e., supervisor-level) or
\emph{unprivileged} (i.e., user-level) states. Generally,
$\textsl{privileged}$ corresponds to the value of a bit in a system
register like \texttt{flags} in Figure~\ref{fig:a-tiny-kernel}.

Transitions between unprivileged states cannot change the values in the system
registers. Moreover when a state is unprivileged, the only way for its successor
to be privileged is by performing a interrupt transition, detailed below.

\mysubparagraph{Transitions} %
The transitions (regular and interrupt) have already been described in
Section~\ref{sec:regular-interrupt-transitions}. In a (standard,
Von Neumann) machine execution, getting the next instruction
correspond to fetching in memory the opcode pointed by the program counter, then decoding it:%
\[ \textsl{next}(s)\enspace\triangleq\enspace\textsl{decode}(s.\textsl{mem}[s.\textsl{regs}[\texttt{pc}]]) \]

We do not detail the format of instructions, as it is standard; in our
tool it is encoded as a sequence of basic instructions of DBA
intermediate language of the BINSEC tool \cite{bardin2011bincoa}.

The \emph{interrupt transition} $\overset{\scriptscriptstyle \textsl{interrupt}}{\to}$
corresponds to the reception of a hardware or software interrupt.
This transition 1. makes $s'$ privileged, 2. changes
$s'\!.\textsl{regs}[\texttt{pc}]$ to a
specific label that we call the \emph{kernel entry point}, and
3. possibly performs other operations, such as saving the values of
registers in system registers or memory.
% For simplicity we assume that an interrupt transition
% $s \overset{\scriptscriptstyle \textsl{interrupt}}{\to} s'$ may only be performed when $s$ is unprivileged.

\mysubparagraph{Memory protection} %
A key component of any OS kernel invariant consists in ensuring that the
memory protection is properly set up. Indeed memory protection and
hardware privilege are the two mechanisms that the kernel must use to
protect itself from the user code, and a vulnerable memory protection
generally leads to a possible privilege escalation.

The memory protection is modeled as follows. A predicate
$\textsl{accessible}: \mathbb{S} \times \mathbb{A} \to
\mathcal{P}\{R,W,X\}$ returns the access rights for an address $a$ in
a state $s$. A state can perform a regular transition
only if it can access all the memory addresses that are needed for
fetching and executing this instruction; otherwise it must perform an
interrupt transition.

The \textsl{accessible} predicates generalizes every memory protection
mechanism found in usual hardware (note that memory translation can be
encoded if needed in the semantics of instructions). In some systems
(e.g., based on Memory Protection Units), \textsl{accessible} depends
only on the value of some system registers. This means that
unprivileged code cannot change the set of accessible addresses
directly (an indirect attack is still possible, by having the kernel
load corrupt data in system registers). In other systems (e.g., based
on Memory Management Units), \textsl{accessible} also depends on
addresses in memory pointed by a system register (the \emph{memory
  protection tables}, for instance page tables). This makes possible a
direct attack where unprivileged code modifies memory protection
tables directly, if some of the addresses of the memory protection
tables in use are accessible.

\mysubparagraph{Multiprocessor} %
In multiprocessor systems, the state is a tuple of $\mathbb{M} \times
\left(\mathcal{R} \to \mathbb{V}\right)^n$, where $n$ is the number of
processors. For simplicity's sake we focus our explanation on the
single-processor case: the system is assumed to have only one processor unless
where explicitely mentioned.

\subsection{Description of the static analysis}

%\todo{Les sous-sections de la static analysis s'appellent aussi A-B-C-D, ce qui est confus.
%  Peut-être les appeller appendix 1,2,3?}

\label{app:sec:formal-description-sound-static-analysis}

%\todo[inline]{Xavier suggests: explain what is an abstract domain, and then refer to ``constraints''}

%\todo{Faire une overview, et mettre ca après ou en annexe}
%\todo{Reexpliquer qq part en quoi le noyau forme une boucle, et qu'on infère les hypervisor invariants grâce à ça?}

In this section we present (with some simplifications) a formal description of our binary-level static
analyzer. 
% that we have used to prove absence of privilege escalation
%and other security properties of the kernel in our case study. Our
%analysis is based on abstract interpretation
%\cite{cousot1977abstract}, which is known to be difficult on
%executable \cite{reps2010there}.  % Section~\ref{sec:analysis-principles}

% presents the mathematical theory underpinning this analysis, while
% Section~\ref{sec:overview} explains how we borrowed ideas from
% state-of-the-art analyzers to build our analysis.

% \subsection{Overview}

% In a nutshell, our analyzer is based on an abstract interpreter
% \cite{cousot1977abstract} for machine-code. It performs a standard
% flow-sensitive analysis on an intermediate representation of machine
% code. For maximum precision, it is fully context-sensitive (i.e. it re-analyzes a called function at each call site )XXX full inline/infinite call strings.

% It uses state-of-the-art
% techniques in machine code analysis like reduced product between
% signed and unsigned meanings of bitvectors \cite{djoudi2016recovering}
% or online construction of the control-flow graph during the data-flow
% analysis \cite{kinder2009abstract}.

% XXX: state analysis, i.e. concretizes into a set of states.

% XXX: extract type information and address range from ELF information,
% but we do not trust it; rather, our domain checks that the code
% complies to these information.

%\subsection{Background}
%\label{sec:analysis-principles}

%\todo{Probablement redondant avec le general principle}
\mysubparagraph{Background}
When given a transition system $\langle \mathbb{S}, \mathbb{S}_0, {\to} \rangle$,
static analysis by abstract interpretation \cite{cousot1977abstract}
allows to compute a finite representation of a \emph{super-set} of the
reachable states. As a set of states is isomorphic to a property over
states, this finite representation also corresponds to a state
property. Abstract interpretation is \emph{sound}\Ldash{}the computed invariants
are correct by construction\Rdash, and can thus be used as a proof technique.%  It can thus be used as an
% automated proof technique, by checking that the computed property
% implies the property that we want to prove. 

Abstract interpretation works by combining \emph{abstract domains},
i.e. partially-ordered sets of \emph{abstract values}, each representing a set of
\emph{concrete values}. Formally, the \emph{meaning} of an abstract
value $d^\sharp$ belonging to an abstract domain $\mathbb{D}^\sharp$
is given by its \emph{concretization}, which is a function
$\gamma_\mathbb{D}^\sharp: \mathbb{D}^\sharp \to
\mathcal{P}(\mathbb{D})$ mapping abstract value to the set of concrete
values it represent. For instance, intervals $\langle a, b\rangle$ are
finite representations of (possibly infinite) sets of integers,
%$\mathbb{D}^\sharp = \left(\mathbb{Z} \cup  \{-\infty,+\infty\}\right)^2$ and
e.g.,
$\gamma\left(\langle 3,+\infty\rangle\right) = \{ x \in \mathbb{Z}
\mid 3 \le x \}$.

%The computation of a state invariant using abstract interpretation
% on our model $\langle \mathbb{S}, \mathbb{S}_0, {\multipleinstruction}\rangle$
% introduced in Section~\ref{sec:enpowering-attacker}
%follows from the definition of the
%set of reachable state.

The set of reachable states in a
$\langle \mathbb{S}, \mathbb{S}_0, {\to}\rangle$ transition system is
formally defined as follows. We define a function $F$:
%$F: \mathcal{P}(\mathbb{S}) \to \mathcal{P}(\mathbb{S})$ as:
\[ \begin{array}{rcl}
     F &:& \mathcal{P}(\mathbb{S}) \to \mathcal{P}(\mathbb{S})\\
     F(S) &=& S_0 \cup S \cup \{ s'\mid \exists s \in S: s \to s'\} 
   \end{array} \]

\noindent such that the set of states reachable from $S_0$ is defined as the least
fixpoint of $F$ (noted $\mathrm{lfp}(F)$), i.e., is the
smallest set $S$ such that $F(S) = S$.

The computation of an invariant using abstract interpretation mimics
this definition. Given an abstract domain $\mathbb{S}^\sharp$ representing the set of states $\mathbb{S}$, a concretization function $\gamma :
\mathbb{S}^\sharp \to \mathbb{S}$, and a sound approximation $F^\sharp:
{\mathbb{S}^\sharp} \to {\mathbb{S}^\sharp}$ of $F$, i.e.\@ such that
\[ \forall S^\sharp \in {\mathbb{S}^\sharp}:  F(\gamma_{\mathbb{S}^\sharp}(S^\sharp))\ \subseteq\ \gamma_{\mathbb{S}^\sharp}(F^\sharp(S^\sharp))   \]

\noindent then every postfixpoint $P^\sharp$ of $F^\sharp$, i.e. such that \mbox{$F^\sharp(P^\sharp)
\sqsubseteq_{\mathbb{S}^\sharp} P^\sharp$} will be a sound approximation of $\mathrm{lfp}(F)$, i.e.\@:
\[ \mathrm{lfp}(F)\ \subseteq\ \gamma_{\mathbb{S}^\sharp}(P^\sharp)  \]
A postfixpoint of $F^\sharp$ can be computed by upward
iteration sequences with widening \cite{cousot1977abstract}, which consists in
growing the abstract value until it cannot grow any more (which means we found a
postfixpoint).

In general, abstract domains are proved and designed in a modular way
by composition of abstract domains \cite{cousot1979systematic}. The
concretizations functions are used to establish the soundness of the
\emph{transfer functions}, describing how the abstract state is
modified by the $\to$ transition. 

In the remainder of this section we will explain the main abstractions
we use, but will give only informal presentation of the transfer
functions (as they follow from the semantics and from the
concretization), and we will omit the soundness proofs. For the sake
of simplicity, all of our abstract domains concretize into sets of
states $\mathcal{P}(\mathbb{S})$.

% Abstract domains are represented using \emph{semi-lattice}, a
% mathematical structure equiped with an order relation
% $\sqsubseteq^\sharp$ and a join binary operation $\sqcup^\sharp$.

% \subsection{Overview of the analysis}
% \label{sec:overview}

% %XXX: parler aussi de la concurrence, de la représentation intermédaire DBA

% \todo{Faire un dessin avec les domaines abstraits utilisés, et s'en
%   servir pour expliquer l'analyse}

\begin{figure}[t]
  \footnotesize
  \centering
  \begin{tabular}{|l l|}     \hline
    state domain & $\mathbb{S}^\sharp = \mathbb{C}^\sharp \times \mathbb{D}^\sharp$ \\
    control flow domain & $\mathbb{C}^\sharp = \mathcal{P}(\mathcal{P}(\mathbb{L} \times \mathbb{L}))$ \\
    data flow domain & $\mathbb{D}^\sharp = \mathbb{L} \to \mathbb{M}^\sharp$\\
    storage domain & $\mathbb{M}^\sharp = \mathbb{N}^\sharp \times \mathbb{T}^\sharp$ \\    
    numeric domain & $\mathbb{N}^\sharp$ = any conjunction of numeric constraints \\
                  & \qquad over the values bound to $\mathbb{A}_{F}$ and $\mathcal{R}$ \\
    type domain & $\mathbb{T}^\sharp = (\mathbb{A}_{F} \uplus \mathcal{R}) \to \mathcal{T}$ \\
                 & \\
    \multicolumn{2}{|l|}{$\mathcal{T}$ is the set of types, $\mathbb{A}_{F}$ the set of fixed addresses in the kernel,} \\
    \multicolumn{2}{|l|}{$\mathcal{R}$ the set of register names, $\mathbb{L}$ the set of program locations.} \\
    \hline
  \end{tabular}
  \caption{Implementation of the abstract domains.}   \label{fig:domain-implementation}
\end{figure}

\mysubparagraph{Main state abstraction} %
Static analysis of machine code needs to analyse simultaneously
the control flow and the
data flow, as each depends on the other. To this end, our main
abstraction combines a \emph{control-flow abstraction} $\mathbb{C}^\sharp$ with
a \emph{data-flow abstraction} $\mathbb{D}^\sharp$, in a manner similar to
\cite{kinder2009abstract}.

Central to our abstractions is the notion of \emph{program locations}
$\mathbb{L}$.  What is a program location is an implementation choice
of the analyzer: a natural choice is to consider that a program
locations is a kernel address (i.e., a precise address in the kernel
code segment). In our case study we chose a more precise abstraction:
a program location consists in a kernel address with a call stack.
We denote by $\mathscr{L}(s)$ the program location of a state $s$.
%\todo{A(l) has become useless}
% In any case a program location should contain the address pointing to the
% next instruction. For a program location $\ell$, we denote
% $\mathcal{A}(\ell)$ its address.

Then, our \emph{control flow abstraction} $\mathbb{C}^\sharp$ is a graph
between program locations $\mathbb{L}$. This graph is represented as the
set of its edges $\mathcal{P}(\mathbb{L} \times \mathbb{L})$. The meaning
of a graph $c^\sharp$ is that (1) the only reachable program locations
are the nodes in the graph, and (2) the only possible location after
executing a state whose location is $\ell_1$ is a location $\ell_2$ where $\ell_2$ is
a successor of $\ell_1$ in the graph. Formally:%
\begin{multline*}
  \gamma_{\mathbb{C}^\sharp}(c^\sharp) = \{\;s_1 \in \mathbb{S}\ \mid\ 
  \exists\ \langle \ell_1,\ell_2 \rangle \in c^\sharp:\ \exists s_2 \in
  \mathbb{S}: s_1 \to s_2 \\
   \land\;\mathscr{L}(s_1) = \ell_1\;\land\; \mathscr{L}(s_2) = \ell_2 \;\}
%  \land\;s_1.\textsl{regs}[\texttt{pc}] = \mathcal{A}(\ell_1)\;\land\; s_2.\textsl{regs}[\texttt{pc}] = \mathcal{A}(\ell_2) \;\}  
\end{multline*}

% \[ \gamma_{\mathbb{C}^\sharp}(c^\sharp) = \{\;s_1 \in \mathbb{S}\ \mid\ 
%   \begin{array}[t]{l}
%     \exists (l_1,l_2) \in c^\sharp:\ \exists s_2 \in \mathbb{S}:  s_1 \multipleinstruction s_2\;\land\;\\
%     \quad\mathoperator{location}(s_1) = l_1\;\land\; \mathoperator{location}(s_2) = l_2 \;\}
%   \end{array} \]
% 

% Thus the control-flow abstraction can be used to prove that the kernel
% cannot execute unwanted instructions, or unwanted transitions (like in
% return-oriented programming).

The \emph{data-flow abstraction}
$\mathbb{D}^\sharp = \mathbb{L} \to \mathbb{M}^\sharp$ maps each
program location $\ell$ to a \emph{storage abstraction} $m^\sharp$\Ldash{}described hereafter\Rdash{}representing the values in the memory and registers. It is the standard
abstraction for flow-sensitive analyses
\cite{kildall1973unified,cousot1977abstract,rival2007trace}. Its
meaning is that the memory and value of registers for a state with
program location $\ell$ must correspond to what is described in the storage
abstraction. Formally:%
\[ \gamma_{\mathbb{D}^\sharp}(d^\sharp) = \{\ s \in \mathbb{S}\ \mid\enspace
\exists \ell \in \mathbb{L}: \ell = \mathscr{L}(s)\ \land\ s \in \gamma_{\mathbb{M^\sharp}}(d^\sharp[\ell])\ \} \]
%\[ \gamma_{\mathbb{D}^\sharp}(d^\sharp) = \{\ s \in \mathbb{S}\ \mid\ \ s \in \gamma_{\mathbb{M^\sharp}}(d^\sharp[\mathscr{L}(s)])\ \} \]
%\[ \gamma_{\mathbb{D}^\sharp}(d^\sharp) = \{\ s \in \mathbb{S}\ \mid\quad (s.\textsl{mem},s.\textsl{regs}) \in \gamma_{\mathbb{M^\sharp}}(d^\sharp[\operatorname{location}(s)])\ \} \]

The main state abstraction $\mathbb{S}^\sharp$ simply consists in a
product~\cite{cousot1979systematic} of these previous abstractions. It represents a set of
states that must match both abstractions. Formally: 
%i.e.\@ states whit must be a reachable
%kernel instruction whose memory and registers is described by the
%storage abstraction. Formally:
\[ \gamma_{\mathbb{S}^\sharp}(c^\sharp,d^\sharp)\ =\ \gamma_{\mathbb{C}^\sharp}\!(c^\sharp)\;\cap\;\gamma_{\mathbb{D}^\sharp}\!(d^\sharp) \]

The analysis works by performing multiple rounds of the following
steps in sequence: % (each step corresponding to the $F^\sharp$ function):

\begin{enumerate}
\item Perform a standard data-flow analysis using the current CFG
  $c^\sharp$, to compute a new $d^\sharp \in \mathbb{D}^\sharp$.
  
\item Iterate over all locations $l \in c^\sharp$ to compute all
  possible outgoing edges, given the possible memories at the
  instruction entry $d^\sharp[l]$ (this uses the same \texttt{resolve}
  function than \cite{kinder2009abstract}). Newly-discovered edges are
  added to the CFG $c^\sharp$.
\end{enumerate}

The iteration sequence starts with an abstraction $(c^\sharp_0,d^\sharp_0)$ of
the initial states $S_0$. In our case study, this corresponds to
$c^\sharp_0$ containing only the first instruction in the kernel, and
$d^\sharp_0$ representing a mapping from this instruction to every
possible initial values of memory and registers (obtained from the
kernel executable file, plus system registers indicating that the
interrupt received is a \texttt{RESET}). The analysis terminates when
the fixpoint is reached, i.e., no new edge is discovered in the CFG. In
practice, several small optimisations are used to reuse results
between rounds (e.g., caching the results), and to have fewer rounds
(by early exploration of the newly-discovered CFG nodes).

\begin{theorem} If the transfer functions for $\mathbb{M^\sharp}$ are
  sound, the result $s^\sharp_{final}$ of the analysis is a sound
  abstraction of all the reachable states in the system (and thus a
  state invariant).
\end{theorem}
% \begin{proof} Each round make $s^\sharp$ increase monotically, and we
%   stop only when the fixed-point is reached, so by
%   \cite{cousot1977abstract} the result is an over-approximation of the
%   collecting semantics.
% %  Standard for AI (Cousot1977),knapper-tarsky, monotony
% \end{proof}

In our case study we use the checked hypothese that kernel-controlled
code is at a fixed location and is not modified (i.e., we have no
self-modifying nor dynamic loading of code). % These could be handled
% using a more complex definition of locations $\mathbb{L}$
% \todo{Citations giacobazzi, marion, blazy}.
This hypothese is checked by verifying that no memory store can modify a read-only region
(see Figure~\ref{fig:partitioning-memory-regions}), and reporting an
error if this happens.

\mysubparagraph{Memory storage abstraction} %
The memory storage abstraction $\mathbb{M^\sharp}$ represents the
contents of all the storage in the system, i.e., memory cells and
registers. Its concretization has signature:
\[\gamma_{\mathbb{M}^\sharp}:{\mathbb{M}^\sharp} \to \mathcal{P}(\mathbb{S})\]

%(\mathbb{A} \to \mathbb{V}) \times (\mathcal{R} \to \mathbb{V})

\noindent where
$\mathbb{S} = (\mathbb{A} \to \mathbb{V}) \times (\mathcal{R} \to
\mathbb{V})$ is a pair of a memory (map from addresses to values) and
values of registers (map from register names to values).

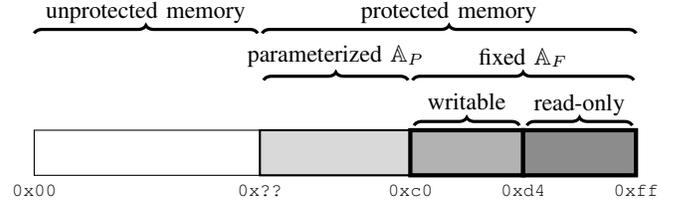
\begin{figure}[htbp]
  \centering

  \begin{tikzpicture}[yscale=0.3]
    \draw[thin] (2,0) rectangle (10,2);
    \draw[fill=black!15,thick] (5,0) rectangle (10,2);
    \draw[fill=black!30,ultra thick] (7,0) rectangle (8.5,2);      
    \draw[fill=black!45,ultra thick] (8.5,0) rectangle (10,2);

    \node[below] at (2,0) {\scriptsize \tt 0x00};
    \node[below] at (5,0,0) {\scriptsize \tt 0x??};        
    \node[below] at (7,0,0) {\scriptsize \tt 0xc0};    
    \node[below] at (8.5,0) {\scriptsize \tt 0xd4};
    \node[below] at (10,0) {\scriptsize \tt 0xff};

    \begin{scope}[yshift=6pt]
      \draw[very thick,decorate,decoration={brace,amplitude=3pt}] (7.05,2) -- (8.45,2) node[yshift=2pt,above,midway] {\small writable};       
      \draw[very thick,decorate,decoration={brace,amplitude=3pt}] (8.55,2) -- (9.95,2) node[above,midway] {\small read-only};
    \end{scope}
    \begin{scope}[yshift=120pt]
      \draw[very thick,decorate,decoration={brace,amplitude=3pt}] (5.05,0) -- (6.97,0) node[yshift=1pt,above,midway] {\small parameterized $\mathbb{A}_P$};
      \draw[very thick,decorate,decoration={brace,amplitude=3pt}] (7.03,0) -- (9.97,0)  node[yshift=1pt,above,midway] {\small fixed  $\mathbb{A}_F$};       
    \end{scope}
    \begin{scope}[shift={(0,-2)},yshift=236pt]
      \draw[very thick,decorate,decoration={brace,amplitude=3pt}] (2,0) -- (4.97,0)  node[yshift=0pt,above,midway] {\small unprotected memory};       
      \draw[very thick,decorate,decoration={brace,amplitude=3pt}] (5.03,0) -- (10,0) node[yshift=0pt,above,midway] {\small protected memory};       
    \end{scope}
    % \begin{scope}[yshift=-4pt]
    %   \draw[thick,decorate,decoration={brace,amplitude=3pt}] (4.95,0) -- (2,0) node[yshift=-2pt,below,midway] {\small user memory};       
    %   \draw[thick,decorate,decoration={brace,amplitude=3pt}] (10,0) -- (5.05,0) node[yshift=-2pt,below,midway] {\small kernel memory};       
    % \end{scope}
  \end{tikzpicture}
  \caption{Partitioning of addresses $\mathbb{A}$}
  \label{fig:partitioning-memory-regions}
\end{figure}

The structure of this abstraction is derived from how we can partition
the memory in the system into the following parts (see
Figure~\ref{fig:partitioning-memory-regions}):

\begin{itemize}
\item \emph{Unprotected memory}: adresses that may be modified by
  non-kernel code.
\item \emph{Protected memory}: protected addresses, i.e., memory
  protection should be set up so that these addresses cannot be
  modified by non-kernel code. This part is sub-partitioned into:
  \begin{itemize}
  \item \emph{Fixed part} ($\mathbb{A}_F$), containing data
    structures whose addresses do not depend on the running
    application (and whose addresses are at a fixed location in the
    kernel executable file). Addresses in this part can be further
    distinguished; in particular \emph{writable} addresses (containing
    the stack, and global variables like \texttt{cur} in the example
    kernel) are distinguished from \emph{read-only} addresses,
    containing the code, jump tables, strings, etc.
  \item \emph{Parameterized part} ($\mathbb{A}_P$), containing the data
    structures whose number and size depends on the application; like
    the circular list of tasks in the example kernel.
  \end{itemize}
\end{itemize}

%Note that the protected memory may encompass more than the kernel
%privileged code, and can include important unprivileged code, for
%instance a device driver or a reference monitor that would run in user
%mode (as found in microkernel-based systems). In most kernels this
%would not be needed to prove absence of privilege escalation, but may
%be useful to prove other important properties.

Following this partition, our main storage abstraction
$\mathbb{M}^\sharp$ is a product
$\mathbb{N}^\sharp \times \mathbb{T}^\sharp$ of two different storage
abstractions corresponding to the different memory parts that we need to track:

\begin{itemize}
\item A precise numerical abstraction $\mathbb{N}^\sharp$, whose
  purpose is to track precisely the numerical values contained in the
  memory cells of the fixed address part, and in the registers;
  obtained by standard~\cite{blanchet2003static} lifting of a
  numerical domain into a domain handling numerical properties of a
  fixed number of memory locations;

\item A typed abstraction $\mathbb{T}^\sharp$ representing the
  contents of the parameters part, and its relation with the
  registers and fixed-address part, detailed in
  Section~\ref{sec:type-abstract-domain}.
\end{itemize}

Note that the unprotected memory is supposed to be suject to arbitrary
modifications from an attacker, so no abstraction needs to keep track
of its possible contents.

$\mathbb{M}^\sharp$ is a standard product of the $\mathbb{N}^\sharp$
and $\mathbb{T}^\sharp$ abstractions, so its concretization is defined
simply as:
\[ \gamma_{\mathbb{M}^\sharp}(n^\sharp,t^\sharp)\ =\ 
  \gamma_{\mathbb{N}^\sharp}(n^\sharp) \cap
  \gamma_{\mathbb{T}^\sharp}(t^\sharp) \]

The above is a simplified formalization of our real abstraction, and
additional extensions are necessary to make it work in practice (see
Section~\ref{sec:static-analysis-machine-code}).

\end{document}